\documentclass[english]{article}
\usepackage{helvet}

\usepackage[T1]{fontenc}
\usepackage[latin9]{inputenc}
\usepackage{geometry}
\usepackage{refstyle}
\usepackage{amsmath}
\usepackage{amsthm}
\usepackage{txfonts}
\usepackage{graphicx}
\usepackage{setspace}
\usepackage[numbers]{natbib}
\usepackage{subscript}
\usepackage{pdfpages}
\makeatletter

\AtBeginDocument{\providecommand\secref[1]{\ref{sec:#1}}}
\RS@ifundefined{subsecref}
  {\newref{subsec}{name = \RSsectxt}}
  {}
\RS@ifundefined{thmref}
  {\def\RSthmtxt{theorem~}\newref{thm}{name = \RSthmtxt}}
  {}
\RS@ifundefined{lemref}
  {\def\RSlemtxt{lemma~}\newref{lem}{name = \RSlemtxt}}
  {}

\numberwithin{equation}{section}
\numberwithin{figure}{section}
  \theoremstyle{plain}
  \newtheorem{prop}{\protect\propositionname}
  \theoremstyle{remark}
  \newtheorem{rem}{\protect\remarkname}
  \theoremstyle{definition}
  \newtheorem{defn}{\protect\definitionname}
  \theoremstyle{plain}
  \newtheorem{lem}{\protect\lemmaname}
  \theoremstyle{plain}
  \newtheorem*{prop*}{\protect\propositionname}

\@ifundefined{date}{}{\date{}}
\usepackage{eucal}
\usepackage{chngcntr}

\makeatother

\usepackage{babel}
  \providecommand{\definitionname}{Definition}
  \providecommand{\lemmaname}{Lemma}
  \providecommand{\propositionname}{Proposition}
  \providecommand{\remarkname}{Remark}

\begin{document}



# Vertex coloring of graphs via phase dynamics of coupled oscillatory networks

Abhinav Parihar[1a], Nikhil Shukla[2], Matthew Jerry[2], Suman Datta[2], Arijit Raychowdhury*[1b]

[1]*Georgia Institute of Technology, Atlanta, GA*

[2]*University of Notre Dame, Notre Dame, IN*

[a]*aparihar6@gatech.edu*

[b]*arijit.raychowdhury@ece.gatech.edu*

## Abstract

While Boolean logic has been the backbone of digital information processing, there are classes of computationally hard problems wherein this conventional paradigm is fundamentally inefficient. Vertex coloring of graphs, belonging to the class of combinatorial optimization represents such a problem; and is well studied for its wide spectrum of applications in data sciences, life sciences, social sciences and engineering and technology. This motivates alternate, and more efficient non-Boolean pathways to their solution. Here, we demonstrate a coupled relaxation oscillator based dynamical system that exploits the insulator-metal *transition* in vanadium dioxide ($VO_2$), to efficiently solve the vertex coloring of graphs. By harnessing the natural analogue between optimization, pertinent to graph coloring solutions, and energy minimization processes in highly parallel, interconnected dynamical systems, we harness the physical manifestation of the latter process to approximate the optimal coloring of k-partite graphs. We further indicate a fundamental connection between the eigen properties of a linear dynamical system and the spectral algorithms that can solve approximate graph coloring. Our work not only elucidates a physics-based computing approach but also presents tantalizing opportunities for building customized analog co-processors for solving hard problems efficiently.

The semiconductor industry is pivoted upon the Von Neumann computer architecture which implements the "Turing Machine" model of computation with a clear distinction between processing units and memory. Computation is carried out through a sequence of instructions with periodic loads and stores to the memory. On the contrary, computation in nature, our brain included, follows a radically different approach. Processing is distributed in all parts of the machine; memory and processors are integrated; clear distinguishable atomic instructions are replaced by continuous time dynamics; and information is encoded in physically meaningful quantities instead of their symbolic interpretations. In spite of the success of the von Neumann computing architecture, its limitations become apparent[1,2] when dealing with certain classes of problems such as associative computing[3,4], optimizations[5], pattern matching and recognition[6]. This has motivated active research in alternative computing models, where dynamical systems have been shown to provide a fundamentally new platform to address these increasingly important problem classes[7–11].

In this paper, we report experimental evidence and the corresponding theoretical foundation for harnessing the continuous time dynamics of a system of coupled relaxation oscillators to solve vertex coloring of a random graphs, a combinatorial optimization problem of large-scale importance. Combinatorial optimizations represent a problem class where an optimal value of a function, or its optimal point, needs to computed within a domain set which is discrete or combinatorial. Vertex coloring of graphs is a combinatorial optimization problem which is NP-hard (non-deterministic polynomial-time hard), unless P=NP. This means that the best algorithms end up searching the whole domain set for at least some problem instances. Vertex

coloring is also one of the most studied NP-hard combinatorial optimization problems not only for its significance in computational theory but also for its many real world applications like fault diagnosis[12], scheduling[13–15], resource allocation[16]; Such problems are believed to be solved, or approximated, efficiently in natural processes because they can explore the solution space in a massively parallel manner[17]. In fact, for any deterministic system to be able to solve such hard problems, be it a sequential deterministic Turing machine or a continuous time dynamical system, exponential resources are required which can be in terms of time, hardware components, maximum magnitudes of variables or their precision[7,18–20].

In the last three decades, dynamical systems as well as hardware implementations have been proposed to solve NP problems, many of which implement some form of algorithm for solving, or approximating, such problems. Important attempts include Quantum computers[21], Cellular automata[22], Hopfield networks[5], Ising model formulations[23], chaotic nonlinear attractor systems[7], iterated projections[24], Memcomputing[25] and stochastic searching using non-repeating phase relations among oscillators[26]. All of these approaches, except iterated maps, are based on the idea of interconnected "nodes" which exchange, store and process information among themselves. One particular variant of such dynamical systems is based on a network of coupled oscillators whose phase and frequency dynamics can be exploited to encode system states[3,27,28]. These computational kernels, which have been claimed to mimic the spiking networks of the human brain, have been successfully used in associative computing, demonstrating significant improvements in energy-efficiency and performance when applied to video analytics[29]. Theoretical models of coupled sinusoidal oscillators, (*Kuramoto* models[30,31]), or van der pol oscillators[32,33] have often been used to theoretically study asymptotic limits; but their physical

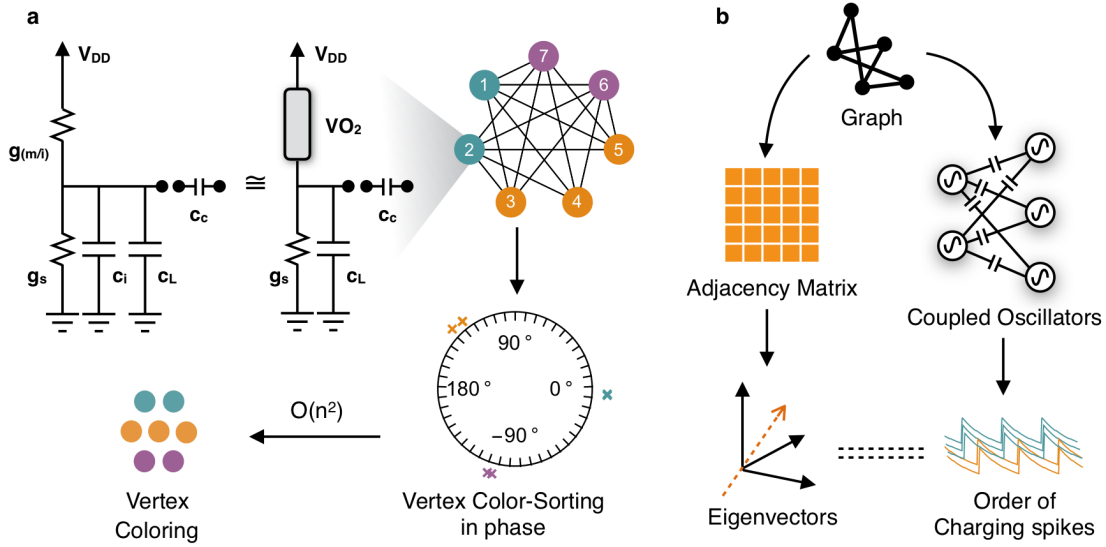

**Figure 1: Overview of the circuit and system dynamics.** (a) Overview of the proposed system for vertex coloring and a simulation example. First step is a coupled relaxation oscillator circuit where the oscillators are composed of a series combination of VO$_2$ device and a resistor (with a loading capacitor in parallel), and are connected in a graph using capacitors. The equivalent circuit diagram of the VO$_2$ oscillator is shown using an internal capacitance $c_i$ and a phase changing conductance $g_{(m/i)}$ which switches between metallic conductance $g_m$ and insulating conductance $g_i$. An example 3-partite graph is simulated and the relative phases of these oscillators are shown in a phase diagram which shows vertex color-sorting in phase, and can be used to calculate vertex-coloring with $O(n^2)$ complexity (b) The circuit is composed of VO$_2$ oscillators capacitively coupled in a network same as the input graph. The final order of phases, or charging spikes, of the oscillators is related to the eigenvectors of the adjacency matrix of the input graph which in turn are related to the solution of the graph coloring problem.

implementations and experimental evidence of such networks to solve computationally hard problems have remained elusive.

In this article, we establish that a system of coupled relaxation oscillators fabricated using Vanadium dioxide (VO$_2$) metal-insulator-transition devices and coupled capacitively, can lead to system dynamics on which vertex coloring of unweighted and undirected graphs (hitherto referred to as the graph coloring problem) can be successfully mapped (Fig. 1a). We demonstrate experimentally and using simulations that when such relaxation oscillators are coupled using only capacitances in a manner topologically equivalent to an input graph, their steady state

4

phases can be used to approximate the solution of the NP-hard *minimum graph coloring problem*. For this, we propose a reformulation of the graph coloring problem where instead of finding a color assignment for each node, the objective is to find a circular ordering or circular permutation of the nodes such that the same colored nodes appear together in the ordering. Such a reformulation preserves the hardness of the problem and is useful for interpreting the output of our circuit (Fig. 1a). We show analytically that the dynamics of such a coupled relaxation oscillator system is intrinsically connected to spectral algorithms for graph coloring[34–36], which use eigenvectors of adjacency matrix of the input graph to approximate the solutions. Alternatively, the permutation of steady state phases of coupled relaxation oscillators depends on eigenvectors of the adjacency matrix in the same way as have been used by spectral algorithms for graph coloring (Fig. 1b). A programmable circuit for a such a coupled oscillator system, where the oscillators are coupled in a graph with adjacency matrix, $A$ and coupling capacitance matrix $C_c$ is shown in Fig. 2. Our simulation results show that the hardness of problem instances has, on average, expected effects on important metrics of solutions found using such a circuit like the number of colors detected and the settling time.

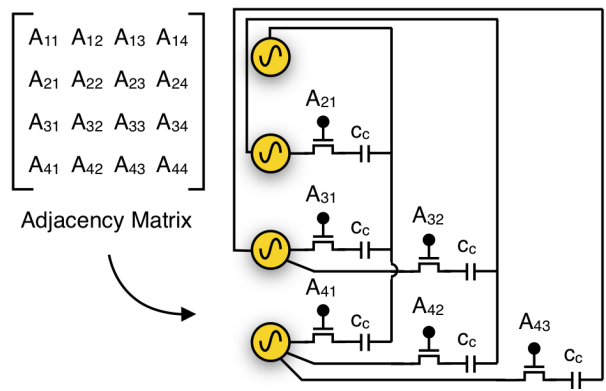

**Figure 2: Coupled oscillator circuit schematic.** A circuit of 4 coupled oscillators with capacitive connections between oscillators controlled using switches corresponding to the adjacency matrix $A$ and coupling capacitance $c_c$. The subscripts denote the corresponding entries in $A$. Note that $A_{ij} = A_{ji}$, $A_{ii} = 0$ and $A_{ij} \in \{0,1\}$.

It is well known that eigen properties of the coefficient matrix in the evaluation equation of a dynamical system determine important structural properties of the system including stability, bifurcation, energy minima(s) and overall system dynamics. In this report, we provide a theoretical bridge and experimental evidence that a dynamical system whose coefficient matrix inherits properties of the incidence matrix of a graph, can indeed emulate spectral graph algorithms just through its time evolution. We envision such dynamical systems to provide foundational paradigms in the development of next-generation computational accelerators and kernels.

**Results**

**A. Minimum Graph Coloring Problem and its reformulation**

The objective of *graph coloring* or *vertex coloring* is to assign one color (out of total $k$ colors) to each vertex of an undirected graph such that no two adjacent vertices receive the same color. A graph coloring that minimizes the number of colors $k$ is called *minimum graph coloring*. The minimum $k$ for which a correct coloring is possible is called the *chromatic number* of the graph. A graph which can be colored using at most $k$ colors is called a *k-partite graph*. A $k$-partition of a set, like the set of nodes, is a grouping of the elements of the set into $k$ groups. Hence, a vertex coloring with $k$ colors is a $k$-partition. We reformulate the objective of finding a color assignment to finding a circular permutation of nodes such that the same colored nodes appear together. We refer to this reformulation as *vertex color-sorting* and the corresponding optimal version as *minimum vertex color-sorting*. Calculating a color assignment from a *color-sorting* is $O(n^2)$. This is because if $P$ is the permutation matrix for the color-sorting and $A$ is the symmetric adjacency matrix of the graph, then the color assignments can be found by observing

the '0' diagonal blocks in the matrix $PAP^T$. This makes *vertex color-sorting* as hard as *vertex coloring* (see Supplementary Section 4). Using this method, any permutation $P$ gives a correct color assignment but a *better* permutation gives lesser number of colors, and an optimal color-sorting permutation gives the minimum number of colors. In the proposed coupled oscillator system, each oscillator represents a vertex (or node) of the graph. Any two nodes connected in the original graph by an edge (as indicated by a '1' in $A$), are capacitively coupled in the hardware implementation. As the coupled system evolves, the relative phases of the oscillators are ordered, and we observe that the relative ordering of the phases approximates *minimum vertex-color-sorting* of the original graph.

**B. Experimental Demonstration of Vertex Coloring in a Coupled Oscillator Network**

We construct a relaxation oscillator by exploiting the electrically induced large and abrupt change in resistance across the insulator-to-metal transition (IMT) in Vanadium Dioxide $(VO_2)$[37,38], and stabilizing it with a negative feedback from a series conductance $g_s$ (details of the single oscillator dynamics have been elucidated in our previous work[28,39]). Since the IMT in $VO_2$ is a materials-level manifestation of hysteretic, resistive threshold switching behavior[40,41] critical to realizing relaxation oscillatory action[42], $VO_2$ based oscillators present a compact, scalable, and potentially low power solution[43] to realizing the fundamental building block of the graph coloring hardware. Further, we use a non-dissipative capacitive coupling scheme to connect the oscillators and achieve frequency synchronization.

7

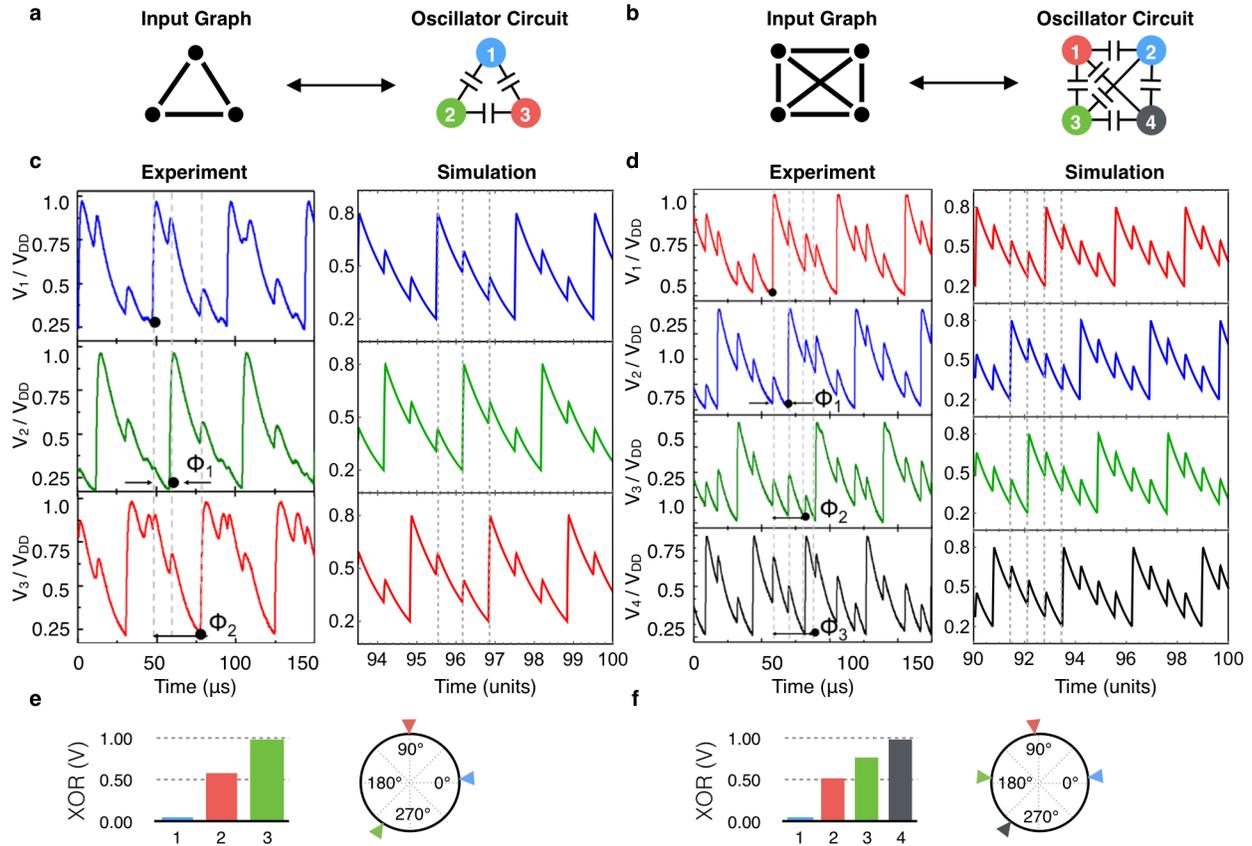

**Figure 3: Phase dynamics of synchronized VO$_2$ based capacitively coupled relaxation oscillators.** (a)(b) Schematics of two representative configurations (a: delta configuration; b; cross-connected ring) of capacitively coupled VO$_2$ based oscillators, and their corresponding graphs. (c)(d) Time domain waveforms from experiment and simulations for the two coupled oscillator configurations in (a)(b), respectively, showing that while the oscillators are synchronized in frequency, no two directly coupled oscillators are in-phase. This important property of the coupled oscillator system enables graph coloring. (e)(f) Time averaged XOR of thresholded outputs of oscillators (each w.r.t. oscillator number 1), and respective polar phase plots showing steady state relative phases detected using PFDs. The XOR values are normalized with respect to the maximum value.

Figure 3 shows two representative configurations of graphs (Fig. 3a: delta configuration; Fig. 3b: cross-connected ring configuration) along with their equivalent implementations using coupled oscillators. The respective time domain waveforms of the oscillators (Fig. 3c,d) reveal a unique relationship among the phases of the oscillators: there is a distinct non-zero phase difference between any two directly coupled oscillators. This is because the nature of capacitive coupling

8

among the relaxation oscillators ensures that two adjacently connected oscillators will tend to force each other out of phase; a rigorous mathematical treatment of the phase dynamics between two coupled $VO_2$ based oscillators has been detailed in an earlier work[39]. Additionally, when an oscillator is connected to multiple other nodes, the net phase of the oscillator is the aggregate of the 'repelling effect' of all the other connected oscillators. As such, in the light of vertex coloring, such a circuit is expected output phases which are clustered by color, i.e. oscillators with the same color have phases which are close together. But such an interpretation of oscillator phases is weak and is difficult to apply for most cases of graphs which either do not have well clustered phases or have incorrect clusters. Our interpretation of outputs as *color-sorting* solves all these problems and is well defined. As will be discussed in the next section, the combined repelling effect in a network of oscillators gives special properties to the order of phases of oscillators in steady state, viz. they approximate *minimum vertex-color-sorting*. As discussed earlier, this steady state ordering of phases is then used to calculate a *vertex-coloring*.

Since the oscillators in this work are non-sinusoidal in nature, the steady state phase differences among the coupled oscillators can be calculated using phase-frequency detectors (PFDs) or the time-averaged XOR metric[43]. The time-averaged XOR measure of any two oscillator outputs is calculated by first thresholding the outputs to binary valued waveforms and then taking the average difference in time of these thresholded waveforms over the complete steady-state periodic orbit of the system. The time-averaged XOR metric is proportional to the absolute value of phase difference between the oscillators, and hence it does not differentiate between lead or lag. Figures 3e,f show the relative phases detected using a PFD (shown using the polar phase plots) and the XOR measures of each oscillator with respect to a common reference oscillator (shown as bar graphs).

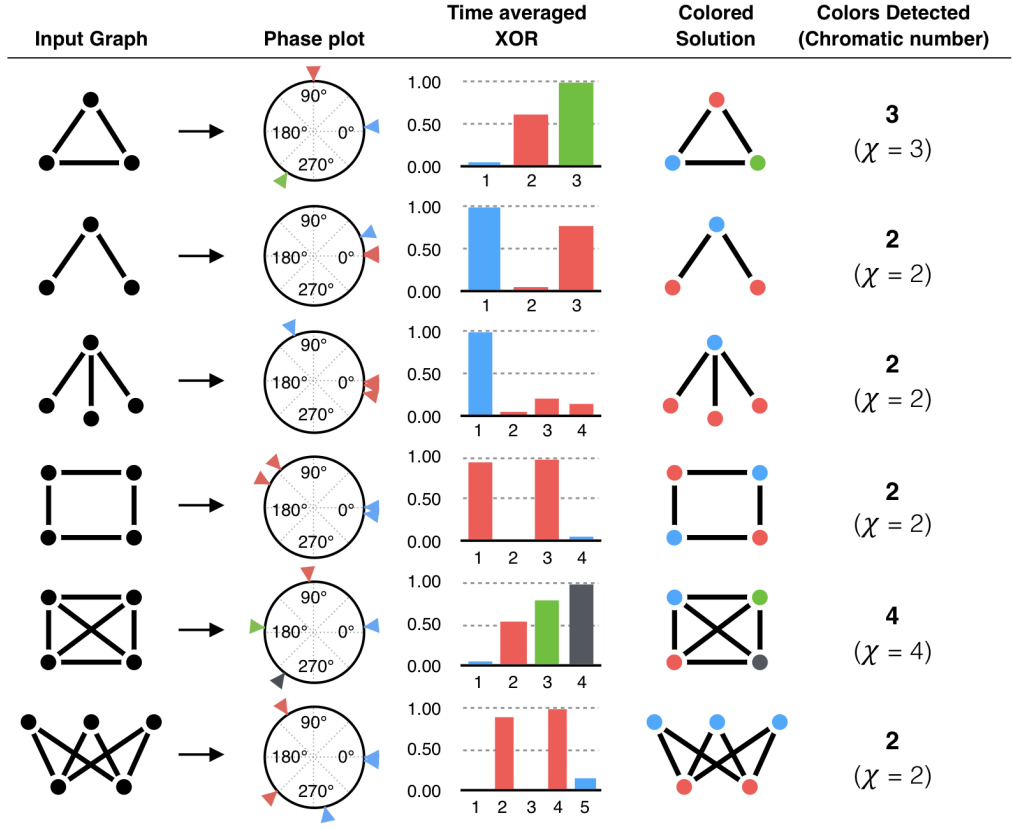

**Figure 4: Graph coloring using the phase dynamics of VO$_2$ based coupled relaxation oscillator system.** Various graph configurations, and their experimentally obtained solutions (PFD outputs and XOR values) using the coupled relaxation oscillator system. After mapping the graphs onto the coupled oscillator hardware, the steady state order of phases of oscillators is used as a color-sorting and a color assignment is calculated.

Next, we experimentally investigate the coloring of some other graph configurations with up to five vertices, using the system of VO$_2$ based oscillators (Fig. 4). The coupled oscillators are configured to represent the respective graphs as discussed earlier, and the corresponding values of the time-averaged XOR along with the respective phase plots of the oscillators are shown in Fig. 4. It can be observed that the hardware is able to *optimally* color all the graphs investigated here.

**C. Analytical Model, Piecewise Linear System Dynamics and connection to Spectral Algorithms**

The mathematical model of the circuit is created as follows. The VO$_2$ devices switch between a low resistance metallic state with conductance $g_m$ and a high resistance insulating state with conductance $g_i$ based on the voltage $v_d$ across their two terminals. On increasing $v_d$ the device switches to a metallic state (insulator-to-metal (IMT) transition) after a threshold $v_h$, and on decreasing $v_d$ below $v_l$ the device switches back to an insulating state (metal-to-insulator (MIT) transition). Here $v_h > v_l$ and $v_h - v_l$ defines the hysteresis in switching. Consider a supply voltage which is applied across the series combination of such a hysteretic device and a conductance $g_s$ where the subscript $s$ denote a series conductance. Without loss of generality, we assume that at $t = 0$ the device is in high resistance state and the voltage drop across the device $v_d = 0$. The internal capacitance of the device charges up and $v_d$ increases and eventually crosses the threshold $v_h$. Due to this the device transitions into a metallic state which causes the internal capacitance of the device to discharge and reduces $v_d$ which finally drops below $v_l$. This causes the device to switch back to the insulating state resulting in oscillations with piecewise linear dynamics. In case of the coupled oscillator circuit, a loading capacitance $c_L$ of appropriate magnitude is required as shown in Fig. 1a for correct circuit operation.

The dynamics of a single oscillator (Fig. 1a) can be written as the following piecewise differential equation:

$$cv'(t) = -g(s)v(t) + p(s) \tag{1}$$

where $c$ is the lumped capacitance of device along with the loading capacitance and parasitics, $s \in \{0,1\}$ is the state of system - charging (denoted by 1) or discharging (denoted by 0), and $g(s)$ is the net path conductance in state $s$, with $g(s) = g_s + g_i s$. If the voltage $v$ is normalized to

$v_{dd}$ then $p(s) = s$. The dynamics of a circuit of identical coupled relaxation oscillators can be described using the following matrix differential equation:

$$\boldsymbol{x}'(t) = (C_i + C_c + C_l)^{-1}[-G(\boldsymbol{s})\boldsymbol{x}(t) + g_i\boldsymbol{s}] \qquad (2)$$

Here, $\boldsymbol{x}$ is the vector of all voltages (normalized to $v_{dd}$), $C_i$ is a diagonal matrix with the diagonal elements equal to the internal capacitances of the corresponding oscillator nodes, $C_c$ is the coupling capacitance matrix with diagonal elements equal to the sum of all the coupling capacitances connected to the corresponding nodes and off-diagonal elements equal to the coupling capacitances of corresponding pair of nodes with negative sign, $\boldsymbol{s} \in \{0,1\}^n$ is the vector of states of all oscillators, $G(\boldsymbol{s})$ is a state dependent diagonal matrix with $diag(G(\boldsymbol{s})) = g_s + g_i\boldsymbol{s}$, and $C_l$ is a diagonal matrix corresponding to the extra loading capacitors. These loading capacitors effectively add to the internal capacitance and are chosen such that $diag(C_c + C_l)$ is constant. When all oscillators have equal internal capacitances $c_i$ and equal coupling capacitances $c_c$, then $C_i = c_iI$ where $I$ is the identity matrix, $C_c = c_cL$, $L$ being the laplacian matrix of the graph with $L = D - A$ where $D$ is a diagonal matrix of degrees of vertices and $A$ is the adjacency matrix of the graph. In such a case, a simple choice of $C_l$ is $C_l = c_c(nI - D)$. These oscillators have very high charging rate (due to the high conductance of the metallic state of the VO$_2$ devices) and low discharging rates (due to relatively high resistance of the pull-down resistors), and as such, the phases of oscillators and their permutation can be read by observing the relative positions of the charging spikes.

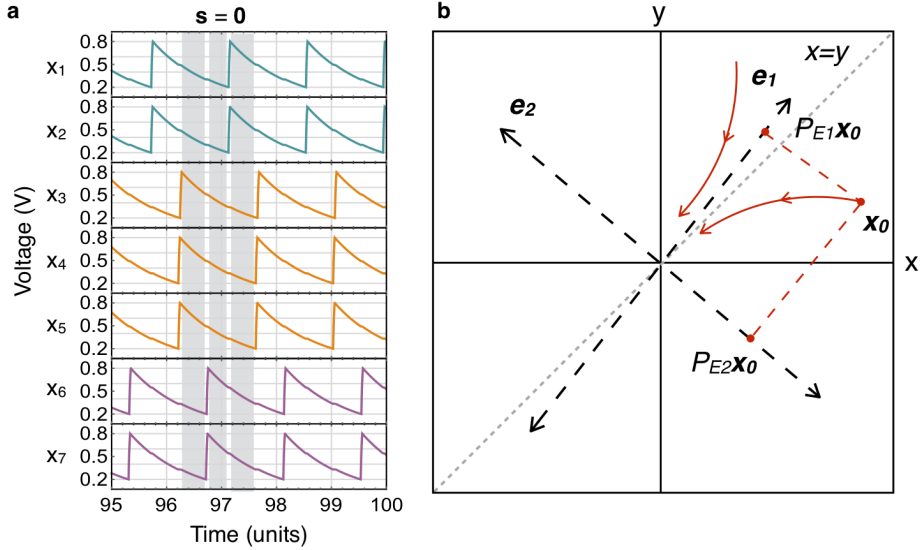

**Figure 5: System dynamics and asymptotic permutation.** (a) Simulation waveforms of a circuit connected in the 3-partite graph shown in figure 1a. The gray regions show the time when the system is in state $s = 0$ (b) Representative figure showing the relation between the asymptotic order of components of the state vector and the eigenspaces in a two-dimensional linear system where the coefficient matrix has negative eigenvalues.

Considering the system of (2), if there exists a limit cycle where the system settles to a certain order of charging spikes, then the order of charging spikes is same as the order of components of the state vector $x$ in state $s = 0$ (Fig. 5a), where the orders are considered unique up to a circular permutation and $0$ is a vector of zeros. In other words, orders which are circular permutations of each other are considered same. The system of (2) in state $s = 0$ is a linear dynamical system with all negative eigenvalues and the asymptotic order of components of the state vector $x$ is determined by the asymptotic direction of the system trajectory. A representative figure of system trajectories in such a linear dynamical system in two dimensions is shown in Fig. 5b in which $e_1$ and $e_2$ are the eigenvectors with distinct negative eigenvalues of the coefficient matrix $B$ of a linear dynamical system $\dot{x} = Bx$, where $x = \{x, y\}$. The eigenspaces $E_1$ and $E_2$ are the lines along $e_1$ and $e_2$ respectively and $P_{E1}$ and $P_{E2}$ are projection matrices for these eigenspaces. Considering (2), the eigenvectors of the coefficient matrix $(C_i + C_c + C_l)^{-1}$ with the least

negative eigenvalues are actually same as the eigenvectors of the adjacency matrix $A$ with most negative eigenvalues (Supplementary Proposition 1). These eigenvectors determine the asymptotic order of components of the state vector $x$ in the following way. Let $T(x)$ represent the order of components of vector $x$. Referring to Fig. 5b, let $x_0$ be the initial starting point of the system, $E_1$ the eigenspace with least negative eigenvalue, $E_2$ with the next higher eigenvalue and so on. Then the asymptotic order of components is same as the order of components of $P_{E1}x_0$. In case $P_{E1}x_0$ has some components which are equal (with respect to Fig. 5b, it means if $e_1$ lies along the $x = y$ line) then the order among those equal components is decided by the order of $P_{E2}x_0$ and so on (see Supplementary Section 2 for analytical derivations). If this operation of combining two orders is represented by $\oplus'$ where the first order is preferred over the second, then the asymptotic order can be written as:

$$Q(x_0) = T(P_{E1}x_0) \oplus' T(P_{E2}x_0) \ldots \qquad (3)$$

The eigenvectors with least negative eigenvalues determine not just the asymptotic order in state $s = 0$ but also the limit cycle. To obtain an intuitive understanding of how the system settles to a limit cycle with the correct color-sorting, i.e. steady state oscillations where order of charging spikes become constant and equal to a correct color-sorting, we note that such a limit cycle exists if:

1. The state $s = 0$ does not change $T(x)$ when $T(x)$ is a correct color-sorting.
2. The charging states $s \neq 0$ change $T(x)$ by only a circular permutation.
3. The state $s = 0$ does not change $T(x)$ even when $T(x)$ is a circular permutation of the correct color-sorting for which property 1 holds.

Property 1 is true when the system is at state $x$ such that $T(x)$ corresponds to the order determined by the least negative eigenvalues (from (3)) because then it is same as the asymptotic order of components and hence does not change. Also, the eigenvectors with the least negative eigenvalues have a property, as used in spectral algorithms for graph coloring[34–36], that their components corresponding to same color tend to be equal (or close) for dense graphs, and these components disperse as graphs become sparse. For instance, in a complete 3-partite graph with arbitrary partition, the eigenvectors have exactly equal components for nodes belonging to the same partition subset (color class). Hence, the order determined by these least negative eigenvectors (or eigenspaces) as given by (3) will correspond to a correct color-sorting with minimum number of colors for dense graphs and the number of colors would increase for sparser graphs. This can also be seen in the light of perturbation theory of matrices. Any sparse $k$-partite graph (or equivalently its adjacency matrix) can be obtained from a complete $k$-partite graph with the same partition structure by removing some edges (or changing '1's in its adjacency matrix to '0's). Such a "perturbation" of adjacency matrix has an equivalent effect of rotating its eigenvectors by some angles[44] which depend on the magnitude of perturbation, which in this case is the number of edges removed. We empirically find properties 1-3 to be true for complete $k$-partite graphs which are the densest $k$-partite graphs (see Supplementary Section 5 for an analytical discussion), and considering sparse graphs as perturbations of complete partite graphs, we can say that vertex color-sorting using the coupled relaxation oscillator circuit becomes less optimal as graphs become sparse. This is known to be true for spectral and other coloring algorithms as well that $k$-colorable dense graphs are easier to color than sparser ones.

**D. Simulation Results and Performance Assessment**

We simulate the dynamical system as described by (2) for random graph instances of 3-colorable graphs. The initial conditions are chosen at random $x_0 \in [0.35, 0.65]^n$ and $s = 1$ at $t = 0$ because all oscillators are in charging state when the power is switched on. Without loss of generality, $v_l$ and $v_h$ are chosen as 0.2 V and 0.8 V respectively with a supply voltage of 1V. We use a random graph generation model $G(n, k, i)$ to generate instances of colorable $k$-partite graphs with total $n$ nodes. The graphs are generated by first choosing a random $k$-partition of $n$ nodes, then creating a complete $k$-partite graph with this $k$-partition and finally removing random $i$ number of edges from this complete graph. *Average connectivity* is defined as the ratio of total number of edges in the generated graph $G(n, k, i)$ to the total number of edges in the complete $k$-partite graph with the same partition.

As is true with hard problems, even in graph coloring problems no heuristic algorithm works best for all graph instances[45]. Also different heuristics work better for different instances, and hence no single *order parameter* can account for the hardness of an instance of a graph coloring problem[46–48]. The most commonly used order parameter is average connectivity[49]. We use this parameter to account for the hardness of the problems being solved and observe how a coupled relaxation oscillator network behaves for problems with varying levels of average connectivity. Observations are made particularly about the *cluster diameter*, the *number of colors detected* and the *settling time*, which are defined as follows.

When the coupled oscillator circuit settles to a correct color-sorting, the phases of oscillators or nodes with the same color form a cluster for many graphs, esp. the dense graphs. The maximum phase difference of two oscillators in the same cluster, i.e. with the same color, is called the

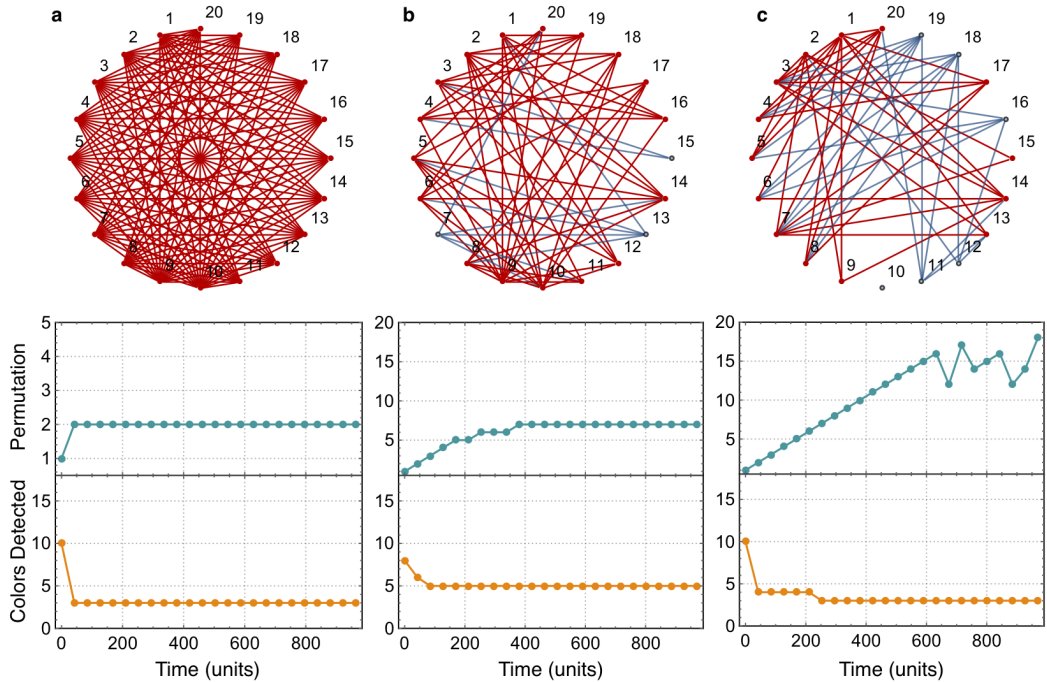

**Figure 6: Effect of hardness on circuit behavior.** (Above) Graph diagram and (below) waveforms of permutation number (a unique number corresponding to each permutation) of charging spikes and number of colors detected with time. All graphs are 3-partite with partition (8,2,10) with varying levels of average connectivity: (a) – 1.0, (b) – 0.57 and (c) – 0.48. All triangles in the graph are shown with red edges and the rest edges are shown in blue.

*cluster diameter*. The *number of colors detected* is calculated using the order of charging spikes at the end of the finite time period for which the circuit is simulated. *Settling time* is the defined as the time after which the number of colors detected does not change till the end of the simulation time.

Figure 6 gives a visualization to how the order of charging spikes evolves with time for 3 different graphs of 20 nodes with decreasing average connectivity. All three graphs are 3-partite with partition (8,2,10). For a single simulation instance, we note the order of charging spikes at various time instances and associate a unique number (within a simulation instance) to each permutation. A plot of this permutation number with time shows how the order, or permutation, of the charging spikes evolve with time. Figure 6 shows this plot along with plot of the number

of colors detected using the order of charging spikes at various times. Figure 6a shows the typical case of a complete partite graph where the order of charging spikes settles quickly to a correct color-sorting, and the number of colors detected falls quickly to the minimum number of colors (3 in this case). Figures 5b and 5c show graphs with lower connectivity but the same partition structure. We make two observations. Firstly, even after the number of colors detected settles down, the permutation or order of charging spikes can evolve. Secondly, Fig. 6c shows lower number of colors detected than Fig. 6b but the settling time is higher for Fig. 6c. As such, both settling time and number of colors detected can be considered as *imperfect* order parameters for hardness of graph coloring just like average connectivity.

Figure 7 shows the performance of such network on random graph instances. We generate 3-partite graphs using $G(10,3,i)$, $G(20,3,i)$ and $G(30,3,i)$ with increasing values of $i$. The graphs which become bipartite after removing $i$ edges are discarded. Various metrics to evaluate the circuit output are plotted against average connectivity. We see that as graphs become sparse with decreasing average connectivity, the cluster diameter increases (Fig. 7a). This comparison is made only among those graphs where the final phases are clustered correctly into 3 clusters, i.e. the number of colors detected is 3. When graphs are dense and closer to being complete partite, the possibility of them being optimally colored with 3 colors is high, and the settling time on average is less (Fig. 7b,c). As graphs become sparser, the number of colors detected (Fig. 7b) as well as settling time (Fig. 7c) increase statistically on average. It also follows our intuitive

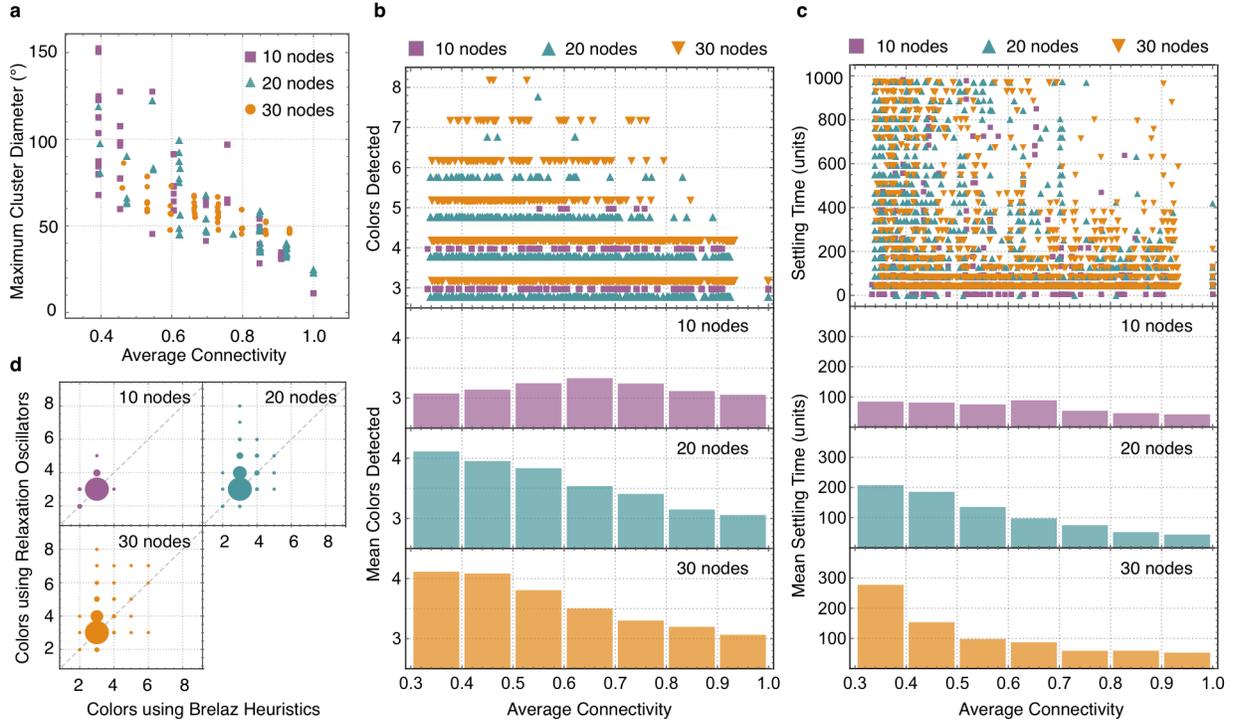

**Figure 7: Simulation results on random graph instances.** (a) Maximum cluster diameter for those graphs for which 3 colors were detected (b) (Top) Number of colors detected plotted against the average connectivity. (Below) Mean colors detected in connectivity intervals (c) (Top) Settling time plotted against the average connectivity. (Below) Mean settling time in connectivity intervals (d) Bubble plot comparison of number of colors detected using circuit simulation of coupled oscillator circuit and using Brelaz heuristics.

understanding that hard computational problems remain hard even under domain transformation, albeit with potential practical implications such as increased energy-efficiency and performance benefits of continuous time systems over their digital counterparts. A comparison of colors detected from simulating the coupled oscillator network with that using Brelaz Heuristics[50] (Fig. 7d) shows the effectiveness of the circuit as a tool to approximate the *minimum graph coloring problem*. Number of colors detected by simulating sample graphs from the second DIMACS implementation challenge[51] are shown in Table 1, where for certain instances we note that the dynamical system outperforms heuristic algorithms.

| Graph | Nodes | Chromatic Number | Minimum number of colors found | |
|---|---|---|---|---|
| | | | Brelaz Heuristics | Coupled Oscillator Circuit |
| huck | 74 | 11 | 11 | 12 |
| myciel3 | 11 | 4 | 4 | 4 |
| myciel4 | 20 | 5 | 5 | 5 |
| myciel5 | 47 | 6 | 6 | 6 |
| myciel6 | 95 | 7 | 7 | 8 |
| david | 87 | 11 | 11 | 13 |
| queen5_5 | 25 | 5 | 7 | 6 |
| queen6_6 | 36 | 7 | 10 | 12 |
| queen7_7 | 49 | 7 | 12 | 12 |
| queen8_8 | 64 | 9 | 15 | 14 |
| DSJC125.1 | 125 | - | 8 | 9 |
| DSJC125.5 | 125 | - | 24 | 34 |

**Table 1: Comparison with Brelaz heuristics.** Comparison of the number of colors detected using Brelaz heuristics with those detected using a coupled relaxation oscillator circuit for various graph instances from the second DIMACS implementation challenge.

In this article, we have established that a system of capacitively coupled relaxation oscillators can perform graph coloring, which is a commonly studied and practically useful combinatorial optimization problem. Further, the connection between system dynamics and the order of steady state phases of oscillators with spectral techniques for graph coloring has been discussed and it shows an innate, yet, hitherto unexplored, connection between the time evolution of dynamical systems and computationally hard problems that have solutions or approximations in the spectral domains.

## Methods

**Experiments with $VO_2$ devices**

*Growth*

The $VO_2$ films have a thickness of 10 nm, and are epitaxially grown on (001) $TiO_2$ using reactive oxide molecular beam epitaxy. The epitaxial mismatch between $VO_2$ and $TiO_2$ results in a tensile biaxial strain of -0.9%.

*Two-terminal $VO_2$ device fabrication*

The electrodes are patterned using contact lithography followed by electron beam evaporation of Pd/Au (20 nm/80 nm) and lift-off in RemoverPG at 70°C. Next, the channel width and isolation are defined by electron beam lithography followed by a CF4 dry etch. Finally, the resist is stripped with RemoverPG at 70°C.

**Circuit simulations of coupled relaxation oscillators**

The oscillator circuits were simulated in Mathematica 10.2 for a finite time (1000 time units). Simulations were performed using default settings for NDSolve routines. The metal-insulator transition events were detected using the inbuilt Mathematica event detection in NDSolve routines with default settings. For Brelaz heuristics, Mathematica routine for Brelaz heurstics was used.

2. Zhirnov, V. V., Cavin, R. K., Hutchby, J. A. & Bourianoff, G. I. Limits to binary logic switch scaling - a gedanken model. *Proc. IEEE* **91,** 1934–1939 (2003).
3. Shukla, N. *et al.* Pairwise coupled hybrid vanadium dioxide-MOSFET (HVFET) oscillators for non-boolean associative computing. in *Electron Devices Meeting (IEDM), 2014 IEEE International* 28.7.1–28.7.4 (2014). doi:10.1109/IEDM.2014.7047129
4. Hoppensteadt & Izhikevich. Synchronization of laser oscillators, associative memory, and optical neurocomputing. *Phys. Rev. E Stat. Phys. Plasmas Fluids Relat. Interdiscip. Top.* **62,** 4010–4013 (2000).
5. Hopfield, J. J. & Tank, D. W. 'Neural' computation of decisions in optimization problems. *Biol. Cybern.* **52,** 141–152 (1985).
6. Chua, L. O. & Yang, L. Cellular neural networks: applications. *IEEE Trans. Circuits Syst.* **35,** 1273–1290 (1988).
7. Ercsey-Ravasz, M. & Toroczkai, Z. Optimization hardness as transient chaos in an analog approach to constraint satisfaction. *Nat. Phys.* **7,** 966–970 (2011).
8. Nikonov, D. E. *et al.* Coupled-Oscillator Associative Memory Array Operation for Pattern Recognition. *IEEE J. Explor. Solid-State Comput. Devices Circuits* **1,** 85–93 (2015).
9. Welser, J. J., Bourianoff, G. I., Zhirnov, V. V. & Cavin, R. K. The quest for the next information processing technology. *J. Nanoparticle Res.* **10,** 1–10 (2007).
10. Srinivasan, G., Sengupta, A. & Roy, K. Magnetic Tunnel Junction Based Long-Term Short-Term Stochastic Synapse for a Spiking Neural Network with On-Chip STDP Learning. *Sci. Rep.* **6,** 29545 (2016).
11. Sharad, M., Augustine, C. & Roy, K. Boolean and non-Boolean computation with spin devices. in 11.6.1–11.6.4 (IEEE, 2012). doi:10.1109/IEDM.2012.6479026
12. Akers, S. B. Fault Diagnosis as a Graph Coloring Problem. *IEEE Trans. Comput.* **23,** 706–713 (1974).
13. Leighton, F. T. A graph coloring algorithm for large scheduling problems. *J. Res. Natl. Bur. Stand.* **84,** 489–506 (1979).
14. Zufferey, N., Amstutz, P. & Giaccari, P. Graph colouring approaches for a satellite range scheduling problem. *J. Sched.* **11,** 263–277 (2008).
15. Xizheng, Z. & Yaonan, W. New mixed broadcast scheduling approach using neural networks and graph coloring in wireless sensor network. *J. Syst. Eng. Electron.* **20,** 185–191 (2009).
16. Woo, T.-K., Su, S. Y. W. & Newman-Wolfe, R. Resource allocation in a dynamically partitionable bus network using a graph coloring algorithm. *IEEE Trans. Commun.* **39,** 1794–1801 (1991).
17. Yan, L. *et al.* Some massively parallel algorithms from nature. *Wuhan Univ. J. Nat. Sci.* **7,** 37–46
18. Vergis, A., Steiglitz, K. & Dickinson, B. The complexity of analog computation. *Math. Comput. Simul.* **28,** 91–113 (1986).
19. Siegelmann, H. T. Computation Beyond the Turing Limit. *Science* **268,** 545–548 (1995).
20. Ercsey-Ravasz, M. & Toroczkai, Z. The Chaos Within Sudoku. *Sci. Rep.* **2,** (2012).

## Acknowledgements

## Author Contributions

AP worked on the development of the theory, simulation frameworks and mathematical models. NS and MJ worked on the experiments. AR advised AP and participated in the problem formulation. SD advised NS and MJ and also participated in the design of experiments and problem formulations.

## Competing Financial Interests Statement

None of the authors have any competing financial interest.



\title{Vertex coloring of graphs via phase dynamics of coupled oscillatory
networks\\
(Supplementary Text)}

\author{Abhinav Parihar, Nikhil Shukla, Matthew Jerry, Suman Datta, Arijit
Raychowdhury}
\maketitle

\section*{Notations}
\begin{itemize}
\item Scalars and vectors are denoted by lower case variables.
\item Matrices are denoted by upper case variables.
\item Single subscripts denote indices for vectors and corresponding columns
for matrices.
\item Double subscripts denote corresponding elements for matrices.
\item General results about the asymptotic order are proved using $x$ as
the state vector. In the context of the paper, the system being described
is the relaxation oscillator system and the state vector $x$ refers
to the output voltage $v(t)$.
\item The state vector representing states of all oscillators is denoted
by lower case $s$ and the diagonal matrix constructed using the state
vector as diagonal is denoted by upper case $\hat{S}$.
\end{itemize}

\section*{Summary}

Following sections describe the proposed coupled relaxation oscillator
system in detail.
\begin{itemize}
\item Section 1 describes the piecewise linear dynamics of a system of a
coupled relaxation oscillators.
\item Section 2 focusses on dynamics in the particular discharge state $s=0$
and explains its relevance and the relationship between eigenvectors
of the coefficient matrix and the asymptotic order of components of
the state vector $x$ in the discharge state $s=0$. 
\item Section 3 discusses similar arguments in other states $s\neq0$.
\item Section 4 explains the reformulation of vertex coloring as vertex
color-sorting.
\item In section 5 we discuss the existence of a periodic cycle in the case
of complete partite graphs with equal nodes in each class of the partition.
The current system can provide the correct, albeit non-optimal coloring
for sparse graphs.
\item In section 6 we give reasons for extending such arguments to general
graphs and why the system moves away from the conditions as graphs
become sparser.
\item Section 7 describes necessary background for the experimental implementation
of such coupled relaxation oscillators using VO\textsubscript{2}
(Vanadium Dioxide) devices.
\item The Appendix contains some results useful for analyses in Section
5.
\end{itemize}

\section{\label{sec:intro}Dynamics of a system of coupled relaxation oscillators}

\begin{figure}
\begin{centering}
\includegraphics{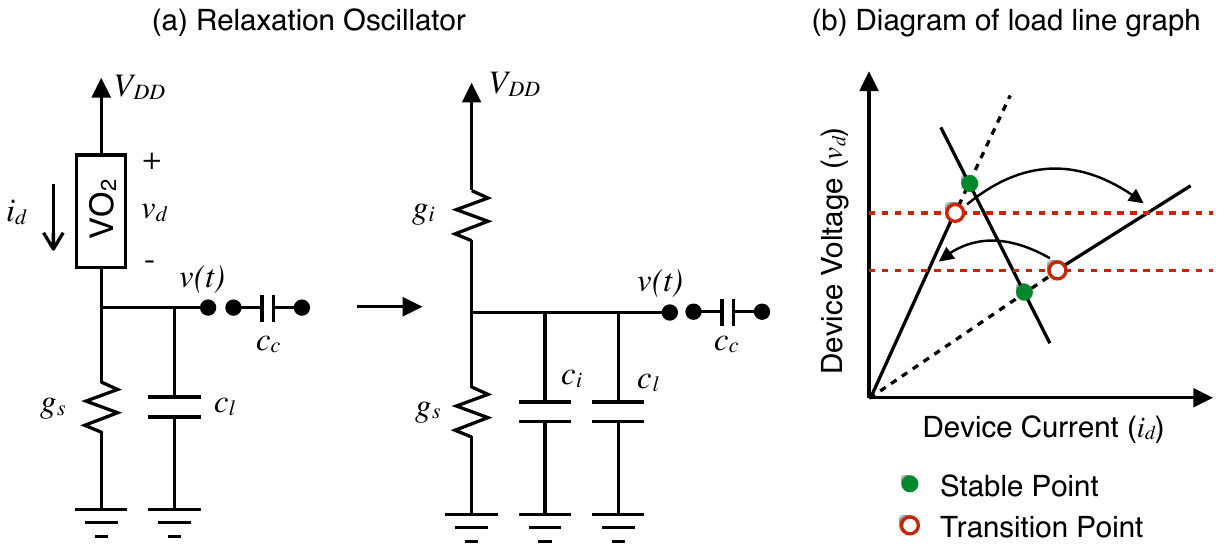}
\par\end{centering}
\caption{(a) A relaxation oscillator circuit and its equivalent circuit in
terms of intrinsic conductance and capacitance. (b) Load line graph
and I-V curve of the device showing transition points, stable points
and oscillations due to hysteresis.\label{fig:circuit}}
\end{figure}

We consider a system of $n$ coupled $VO_{2}$ oscillators, where
each oscillator is a series combination of a $VO_{2}$ device, and
a parallel combination of a series conductance $g_{s}$ and a loading
capacitance $c_{l}$. The $VO_{2}$ device is an MIT (metal-insulator-transition)
device which switches between a metallic state and an insulating state
depending on the voltage $v$ across it. When $v>v_{h}$ the device
switches to a metallic state, and when $v<v_{l}$ the device switches
to an insulating state. $v_{l}\neq v_{h}$ and there is hysteresis,
i.e. system tries to retain the last state when $v_{l}\leq v\leq v_{h}$.
When a VO\textsubscript{2} device is connected in series with a resistance
of appropriate magnitude, it shows self sustained oscillations. As
can be seen in figure \ref{fig:circuit}b, because the stable points
of the circuit in both the states (metallic and insulating) lie outside
the region of operation, i.e. they are preceded by a transition, the
system never settles to a point. 

The dynamics of the coupled system with $n$ oscillators coupled pairwise
to each other using capacitances can be written as: 
\begin{equation}
\left(C_{i}+C_{c}+C_{l}\right)v'(t)=-G(s)v(t)+H(s)\label{eq:system}
\end{equation}

where $s$ is the state of the system, $s=\{s_{1},s_{2},\cdots,s_{n}\}$,
$s_{k}$ being the state of $k^{th}$ oscillator and $v(t)$ is the
vector of all the output voltages of oscillators.. $C_{i}$ is the
intrinsic internal capacitance matrix and $C_{l}$ is the loading
capacitance matrix. These are diagonal matrices with each element
equal to the corresponding capacitance of the oscillator.
\[
C_{i}=\left(\begin{array}{ccc}
c_{i1} &  & 0\\
 & \ddots\\
0 &  & c_{in}
\end{array}\right),\,C_{l}=\left(\begin{array}{ccc}
c_{l1} &  & 0\\
 & \ddots\\
0 &  & c_{ln}
\end{array}\right)
\]
 where $c_{ik}$ is the internal capacitance and $c_{lk}$ is the
loading capacitance of $k^{th}$ oscillator. 

$C_{c}$ is the coupling capacitance matrix
\[
C_{c}=\left(\begin{array}{cccc}
\sum & -c_{c_{12}} & \cdots & -c_{c_{1N}}\\
-c_{c_{21}} & \sum &  & -c_{c_{2N}}\\
\vdots &  & \ddots\\
-c_{c_{N1}} & -c_{c_{N2}} &  & \sum
\end{array}\right)
\]
where $c_{c_{ij}}$ is the coupling capacitances between $i^{th}$
and $j^{th}$ oscillators, and $\sum$ represent the sum of rows (or
columns). When all the coupling capacitances are equal to $c_{c}$,
then $C_{c}$ is basically the scaled Laplacian matrix $L$ of the
graph with $C_{c}=c_{c}L=c_{c}(D-A)$ where $D$ is the diagonal matrix
of degrees of vertices and $A$ is the adjacency matrix of the graph.
It should be noted that the loading capacitances are chosen such that
$diag(C_{c}+C_{l})$ is constant. We envision a system where the oscillators
are connected in a graph which is topologically equivalent to the
input graph. As such the coupling matrix is programmed by the incidence
matrix of the input graph, For each row $i$ in $C_{c}$ every absent
edge $ij$ in the graph adds a loading capacitance of magnitude $c_{c}$
to the $i^{th}$ node to maintain a constant $diag(C_{c}+C_{l})$.
This ensures equal loading effect for all the nodes and symmetric
dynamics. 

$G(s)$ and $H(s)$ are state dependent matrices

\[
G(s)=\left(\begin{array}{ccc}
g_{1}(s_{1}) &  & 0\\
 & \ddots\\
0 &  & g_{N}(s_{2})
\end{array}\right),H(s)=\left(\begin{array}{c}
h_{1}(s_{1})\\
\vdots\\
h_{N}(s_{N})
\end{array}\right)
\]
where
\[
g_{k}(s_{k})=\begin{cases}
g_{ik}+g_{sk} & s_{k}=1,(charging)\\
g_{sk} & s_{k}=0,(discharging)
\end{cases}
\]
and
\[
h_{k}(s_{k})=\begin{cases}
g_{ik} & s_{k}=1,(charging)\\
0 & s_{k}=0,(discharging)
\end{cases}
\]
with $g_{ik}$ and $g_{sk}$ being the internal conductance and the
series conductance of the $k^{th}$ oscillator respectively.

This can be written as:
\[
v'(t)=\left(C_{i}+C_{c}+C_{l}\right)^{-1}\left[-G(s)v(t)+H(s)\right]
\]
where voltages are normalized to $V_{DD}$. In rest of the text, the
state vector will be represented by $x(t)$ instead of $v(t)$.

\subsection{A symmetric system with identical oscillators}

Let us first consider a symmetric system, i.e. equal internal capacitances
($c_{i}$), coupling capacitances ($c_{c}$), internal conductances
$(g_{i})$ and series conductances $(g_{s})$. In such case, $\left(C_{i}+C_{c}+C_{l}\right)=(c_{i}I+c_{c}D-c_{c}A+C_{l})$
where $A$ is the adjacency matrix of the graph and $D$ is the diagonal
matrix of degrees of vertices. One simple choice of $C_{l}$ is $C_{l}=c_{c}(nI-D)$
which makes 
\begin{eqnarray*}
diag(C_{c}+C_{l}) & = & diag(c_{c}D-c_{c}A+c_{c}nI-c_{c}D)\\
 & = & diag(c_{c}nI)\\
 & = & c_{c}n\,diag(I)
\end{eqnarray*}

which is constant. Hence the coefficient matrix becomes 
\[
-G(s)(c_{i}I-c_{c}A+c_{c}nI)^{-1}=G(s)\left(c_{c}A-(c_{i}+c_{c}n)I\right)^{-1}
\]

Let us define $B=\left(c_{c}A-(c_{i}+c_{c}n)I\right)^{-1}$. Also
let $\hat{S}$ be a diagonal matrix where $diag(\hat{S})=s$. Then
$H(s)=g_{i}s$ and $G(s)=g_{s}I+g_{i}\hat{S}$ where $I$ is the identity
matrix. The system of \ref{eq:system} can then be written as:
\begin{equation}
v'(t)=B\left(g_{i}\hat{S}v+g_{s}\left(s-v\right)\right)\label{eq:system-identical}
\end{equation}

We note two important features about the charging transitions: (a)
charging processes are very fast compared to the period of oscillations
(figure \ref{fig:exp_sims_4}), which we also refer to as ``charging
spikes'' and (b) Charging of one oscillator has weak (but finite)
effect on the other oscillators. Hence, we study the dynamics of coupled
relaxation oscillator system in terms of two distinct interacting
systems - the linear dynamics in the discharging state $s=0$, and
the charging transitions.

\begin{figure}
\begin{centering}
\includegraphics{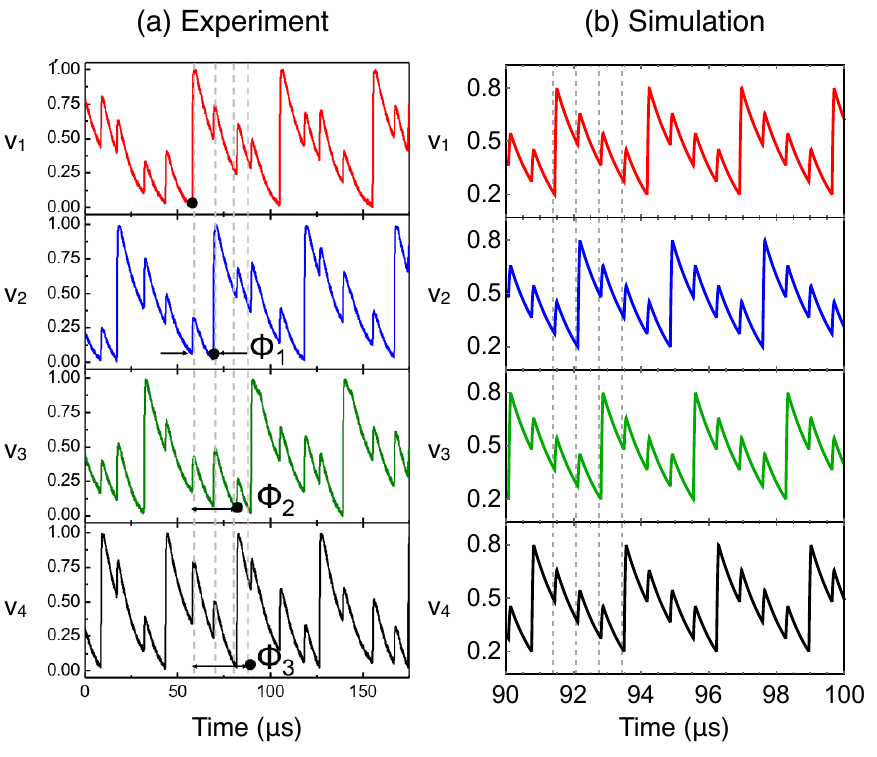}
\par\end{centering}
\caption{Experimental (a) and simulated (b) waveforms of a coupled relaxation
oscillator circuit connected in a complete graph with 4 nodes.\label{fig:exp_sims_4}}
\end{figure}

As the charging processes are very fast, the relative phases of oscillators
are same as the relative times of the charging spikes in the oscillator
waveforms. This gives a good way to visualize how the relative phases
of oscillators evolve with time. For all oscillators, we first note
all the time instants when the charging spikes start. The time differences
between consecutive charging spikes should settle to a constant value
if the oscillators settle, say $\Delta t_{i}$ for the $i^{th}$ oscillator.
If all the oscillators synchronize to a common frequency then $\Delta t_{i}=\Delta t_{0}$
for all $i$. Then at any $n^{th}$ charging spike which occur at
time instant $t_{n}$, we can calculate the relative phase of an oscillator
w.r.t. a hypothetical oscillator whose charging spikes occur at regular
intervals of $\Delta t_{i}$ from the start ($t=0$) as:

\[
\phi(n)=(t_{n}-n\Delta t_{i})\frac{2\pi}{\Delta t_{i}}\,\,(mod\,2\pi)
\]

When all $\Delta t_{i}$ are equal, i.e. the oscillators synchronize,
$\phi(n)$ calculates the relative phases w.r.t. a common $\Delta t_{0}$
for all oscillators. We plot $\phi(n)$ vs $n$ for all oscillators
in figure \ref{fig:phase-converge}. What we observe is that the phases
$\phi(n)$ converge and cluster together for dense graphs but as the
graphs become sparse, which are considered harder, the the phases
do not converge. In the intermediate region between dense and very
sparse graphs, the phase do converge but they do not cluster together
in groups. In these case our proposed algorithm and reformulation
of vertex coloring is particulalry useful because it does not rely
on the clustering of phases. Our algorithm does an $O(n^{2})$ post-processing
on the steady state order of phases and calculates a color assignment
which is always correct but can have non-optimal coloring, i.e. the
number of colors can be more than the chromatic number.

\begin{figure}
\begin{centering}
\includegraphics{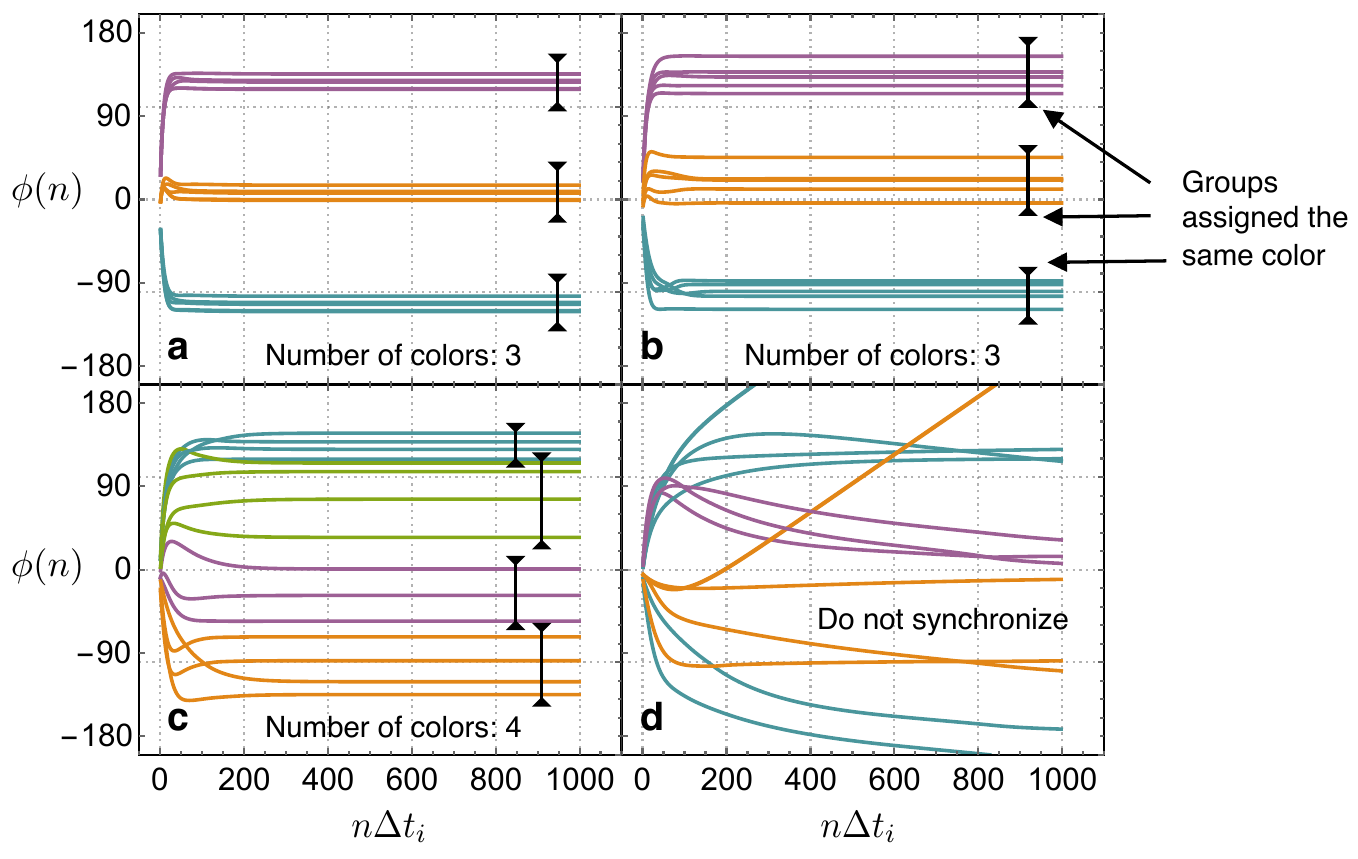}
\par\end{centering}
\caption{The phases $\phi(n)$ plotted against $n\Delta t_{i}$ for four relaxation
oscillator systems for solving 3-colorable graphs with the same color
partition $(5,5,5)$ but with different connectivities. Case (a) is
the case of a complete 3-partite graph, and graphs become sparser
from (a) to (d). The phase clustering degrades as graphs become sparser
and for very sparse graphs (d) the oscillators do not synchronize.
The number of colors detected using our algorithm is shown with each
graph and the nodes which are assigned the same color are indicated.
\label{fig:phase-converge}}
\end{figure}

\section{\label{sec:lds-discharge}Linear dynamics in the discharge phase
$s=0$}

In the state $s=0$ where all the oscillators are in the discharging
state, the system is an autonomous linear dynamical system
\[
x'(t)=-g_{s}\left(c_{i}I+c_{c}L+C_{l}\right)^{-1}x(t)
\]

Hence, the time evolution of this dynamical system is governed by
the spectral properties of the coefficient matrix. In an identical
system, the equation is
\begin{eqnarray*}
x'(t) & = & g_{s}\left(c_{c}A-(c_{i}+nc_{c})I\right)^{-1}x(t)\\
 & = & g_{s}Bv(t)
\end{eqnarray*}

Let the eigenvectors of $B$ be $\mu_{k}$. 
\begin{prop}
\label{prop:eigenvalues-B-A}The eigenvectors of the coefficient matrix
$B$ of the identical system are the same as those of the adjacency
matrix $A$. The eigenvalues $\mu_{k}$ of B are related to the eigenvalues
of $A$ as follows:
\[
\mu_{k}=\frac{1}{c_{c}\left(\lambda_{k}-\frac{c_{i}}{c_{c}}-n\right)}
\]

Moreover, $\mu_{k}<0$ for $1\leq k\leq n$.
\end{prop}
\begin{proof}
For any matrix $M$ with an eigenvalue $m$, the eigenvectors of $M+\alpha I$
and $\beta\left(M+\alpha I\right)^{-1}$ are same as $M$ for any
scalars $\alpha$ and $\beta$. This can be seen as follows: 
\begin{eqnarray*}
\left(M+\alpha I\right)x & = & Mx+\alpha x\\
 & = & \left(m+\alpha\right)x
\end{eqnarray*}

And eigenvectors remain unchanged for matrix inverse. Also eigenvalues
for $\beta\left(M+\alpha I\right)^{-1}$ will be $\beta/(m+a)$. Substituting
appropriate values for $\alpha$ and $\beta$ gives us the required
relation between $\mu_{k}$ and $\lambda_{k}$. 

Now, the Perron-Frobenius theory \citep{horn2012matrixanalysis}implies
that largest eigenvalue of $A$ is less than or equal to the maximum
row sum which is less than $n$, i.e.
\[
\lambda_{max}\leq r_{max}<n
\]
Hence, $\left(\lambda_{k}-\frac{c_{i}}{c_{c}}-n\right)<0$ for all
$k$ which implies that $\mu_{k}<0$ for all $k$. 
\end{proof}

\subsection{Asymptotic trajectories and asymptotic order of components of the
state vector in a linear dynamical system}

In a linear dynamical system with the state variable $x(t)$, the
order of components of $x(t)$ define a permutation at any time instant
$t$. In state $S=0$, the linear dynamical system is
\[
x'(t)=Bx(t)
\]

where $B$ is real, symmetric and the initial state of the system
$x(0)=x_{0}$.

\subsubsection{Geometry of permutation regions}

For any ordering $P$ of components $x_{i1}>x_{i2}>...>x_{in}$, the
region that corresponds to this ordering is given by
\begin{equation}
\mathcal{R_{P}}\left(P\right)=\bigcap_{m=i}^{n}\left(x_{im}>x_{i(m+1)}\right)\label{eq:permutation-regions}
\end{equation}

$\mathcal{R_{P}}(P)$ is a pair of n-dimensional simplexes with one
vertex as the origin and are mirror images of each other about the
origin. As such, any line that passes through the origin either passes
through both of them, or none.

\subsubsection{Asymptotic direction of trajectories}

In a linear dynamical system, the asymptotic order of components is
hence governed by the asympotic direction in which the system state
converges to.
\begin{prop}
In the linear dynamical system $x'(t)=Bx(t)$, where the coefficient
matrix $B$ is real, symmetric and full-rank, the system trajectory
always converges asymptotically to a particular direction. Moreover,
if the asymptotic direction is given by $d(x_{0},B)$ where $x(0)=x_{0}$,
then $d(x_{0},B)$ lies in the eigenspace of $B$ with the largest
eigenvalue (including the sign) almost everywhere, i.e. when the system
starts from anywhere except on a set of measure 0.\label{prop:asym-eigenspace}
\end{prop}
\begin{proof}
Let $x(t,x_{0})$ be the solution of the dynamical system when the
initial starting state $x(0)=x_{0}$. As the fixed point is 0, the
asymptotic direction $d(x_{0},B)$ to which the system state converges
can be written as
\begin{eqnarray*}
d(x_{0},B) & = & \lim_{t\rightarrow\infty}\frac{x(t)}{\left\Vert x(t)\right\Vert }\\
 & = & \lim_{t\rightarrow\infty}\frac{e^{Bt}x_{0}}{e^{\lambda(x_{0})t}}
\end{eqnarray*}
where $\lambda(x_{0})$ is the Lypunov exponent of the trajectory
starting from $x_{0}$. As $B$ is real and symmetric, all its eigenvalues
are real and the matrix is diagonalizable. Let $B=Q\Lambda Q^{T}$,
where $\Lambda$ is the diagonal matrix with of all eigenvalues. Then
\[
d(x_{0},B)=Q\left(\lim_{t\rightarrow\infty}\frac{e^{\Lambda t}x_{0}}{e^{\lambda(x_{0})t}}\right)Q^{T}x_{0}
\]

Let $\lambda_{1}>\lambda_{2}>...>\lambda_{l}$ be the $l$ distinct
eigenvalues of $B$, and let $E_{k},\,1\leq k\leq l$ be the corresponding
eigenspaces. Now, $\lambda(x_{0})=\lambda_{1}$ for $x_{0}\in\bigoplus_{k=1}^{l}E_{k}\backslash\bigoplus_{k=1}^{l-1}E_{k}$.
This means $\lambda(x_{0})=\lambda_{1}$ almost everywhere, i.e. everywhere
except on a set of measure $0$. Hence
\begin{eqnarray*}
d(x_{0},B) & = & Q\left(\begin{array}{ccccc}
1 & 0 & 0 & \cdots & 0\\
0 & 1 & 0\\
0 & 0 & \ddots\\
\vdots &  &  & 0\\
0 &  &  &  & \ddots
\end{array}\right)Q^{T}x_{0}\\
 & = & \left(q_{1a}q_{1a}^{T}+q_{1b}q_{1b}^{T}+...\right)x_{0}\\
 & = & P_{E_{1}}x_{0}
\end{eqnarray*}

Here, the diagonal elements of the middle matrix are ones only for
the rows corresponding to the eigenvector $\lambda_{1}$, and $q_{1a},q_{1b},...$
are orthogonal vectors that span $E_{1}$. Hence $d(x_{0},B)\in E_{1}$
almost everywhere. In case the largest eigenvalue $\lambda_{1}$ of
$B$ has multiplicity 1, $d(x_{0},B)$ is simply $q_{1}$ a.e. 
\end{proof}

\subsubsection{Asymptotic order of components}

The asymptotic order of components of $x(t)$ is determined by the
permutation region in which $d(x_{0},B)$ lie. Let $T(v)$ denote
the order of components of vector $v$, then $T(d(x_{0},B))=T(P_{E_{1}}x_{0})$
is the asymptotic order of components of $x(t)$. The asymptotic order
becomes a little more complex when $d(x_{0},B)$ lies at the boundary
of two or more permutation regions, i.e. some of the components of
$d(x_{0},B)$ are equal. In such cases, $T(d(x_{0},B))$ is only a
partial order as determined by $d(x_{0},B)$. $T(d(x_{0},B))$ can
be extended to a total order by the asymptotic direction of the system
in the remaining space $E_{2}\oplus E_{3}\oplus...\oplus E_{l}$.
Let us denote this by $d(x_{0}\backslash E_{1})$. Also, let $P_{E_{1}}$
be the projection matrix on $E_{1}$, then
\begin{eqnarray*}
d(x_{0},B\backslash E_{1}) & = & \lim_{t\rightarrow\infty}\frac{\left(I-P_{E_{1}}\right)x(t)}{\left\Vert \left(I-P_{E_{1}}\right)x(t)\right\Vert }
\end{eqnarray*}

Now, $d(x_{0},B\backslash E_{1})\perp d(x_{0},B)$. When $d(x_{0},B)$
is at the boundary of some permutation regions, the disambiguation
among these regions, i.e. ordering among the components which are
equal, is done by $d(x_{0},B\backslash E_{1})$ as it is perpendicular
to $d(x_{0},B)$. Hence, the asymptotic order is determined by both
$d(x_{0},B)$ and $d(x_{0},B\backslash E_{1})$. If $d(x_{0},B\backslash E_{1})$
lie at the boundary of some other permutation regions, then the argument
can be extended in a similar way and the asymptotic order of components
is determined by $d(x_{0},B)$, $d(x_{0},B\backslash E_{1})$ and
$d(x_{0},B\backslash E_{1}\oplus E_{2})$ together, and so on.

The extension of the partial order $T(d(x_{0},B))$ using $T(d(x_{0},B\backslash E_{1}))$
is similar to the ordinal sum $T(d(x_{0},B))\oplus T(d(x_{0},B\backslash E_{1}))$
but a preferential one, i.e. the orders determined by $T(d(x_{0},B))$
are preferred over those determined in $T(d(x_{0},B\backslash E_{1}))$.
Let us denote this operation by the binary operator $\oplus'$ which
acts on an ordered pair of two partial orders and gives another partial
or total order.

The range of $\left(I-P_{E_{1}}\right)$ is $E_{2}\oplus E_{3}\oplus...\oplus E_{l}$.
The dynamics that govern the time evolution of $\left(I-P_{E_{1}}\right)x(t)$
in the space $E_{2}\oplus E_{3}\oplus...\oplus E_{l}$ is simply determined
by the eigenvectors and eigenvalues corresponding to $E_{2},E_{3},...,E_{l}$.
Hence from \ref{prop:asym-eigenspace}, $d(x_{0},B\backslash E_{1})\in E_{2}$.
Specifically, 
\[
d(x_{0},B\backslash E_{1})=\left(q_{2a}q_{2a}^{T}+q_{2b}q_{2b}^{T}+...\right)x_{0}
\]
where $q_{2a},q_{2b},...$ are the eigenvectors corresponding to $\lambda_{2}$.
Extending the argument, we have $d(x_{0}\backslash E_{1}\oplus E_{2})\in E_{3}$
and so on. Hence, we have the following:
\begin{prop}
\label{prop:asym-order-discharge}The asymptotic order of components
of $x(t)$ in the linear dynamical system $x'(t)=Bx(t)$, where the
coefficient matrix $B$ is real, symmetric and full-rank, is determined
by $T(d(x_{0},B))$. In case $d(x_{0},B)$ lies on the boundary of
some permutation regions then $T(d(x_{0},B))$ is a partial order
which can be extended to a total order as $T(d(x_{0},B))\oplus'T(d(x_{0},B\backslash E_{1}))$.
And in case $d(x_{0},B\backslash E_{1})$ lies at some boundary then
the asymptotic order is determined as $T(d(x_{0},B))\oplus'T(d(x_{0},B\backslash E_{1}))\oplus'T(d(x_{0},B\backslash E_{1}\oplus E_{2}))$
. Moreover,
\begin{eqnarray*}
d(x_{0},B) & = & \left(q_{1a}q_{1a}^{T}+q_{1b}q_{1b}^{T}+...\right)x_{0}=P_{E_{1}}x_{0}\in E_{1}\\
d(x_{0},B\backslash E_{1}) & = & \left(q_{2a}q_{2a}^{T}+q_{2b}q_{2b}^{T}+...\right)x_{0}=P_{E_{2}}x_{0}\in E_{2}\\
d(x_{0},B\backslash E_{1}\oplus E_{2}) & = & \left(q_{3a}q_{3a}^{T}+q_{3b}q_{3b}^{T}+...\right)x_{0}=P_{E_{3}}x_{0}\in E_{3}
\end{eqnarray*}
and so on. Hence, the asymptotic order of components is determined
as 
\[
Q_{0}(x_{0})=T(P_{E_{1}}x_{0})\oplus'T(P_{E_{2}}x_{0})\oplus'T(P_{E_{3}}x_{0})\ldots
\]

\begin{figure}
\begin{centering}
\includegraphics{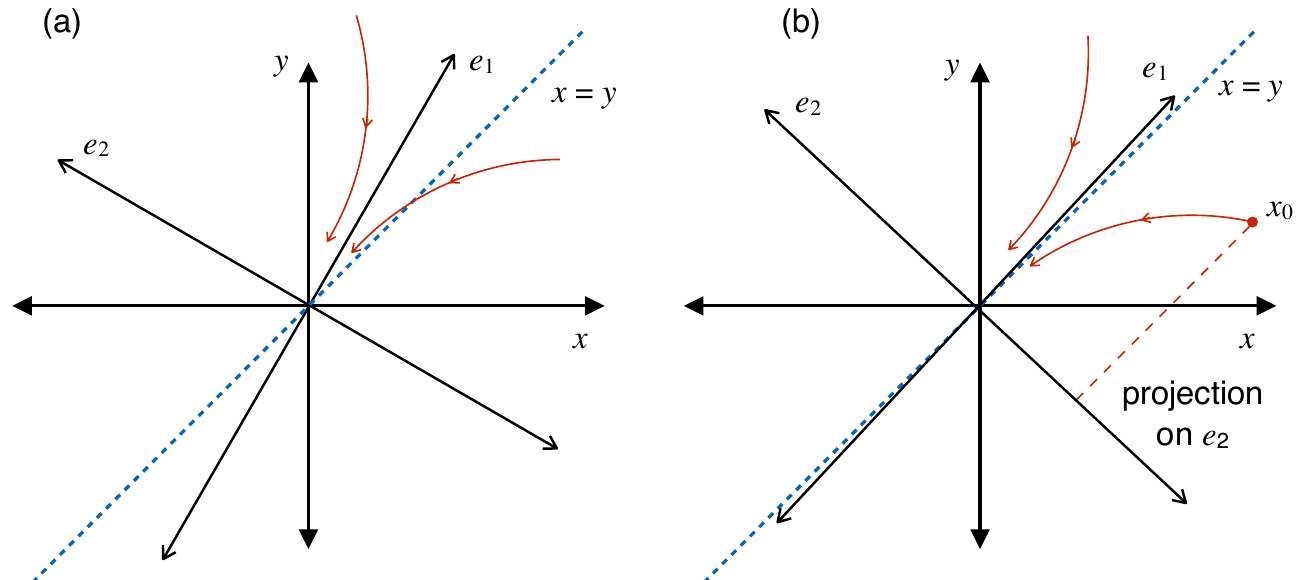}
\par\end{centering}
\caption{Representation of flows in a two dimensional linear dynamical system
where both eigenvalues are negative and $|\lambda_{2}|>|\lambda_{1}|$.
(a) The system trajectory approaches the direction of $e_{1}$ with
time and hence the order of components, i.e. the order of $x$ and
$y$ coordinates is determined by $e_{1}$. (b) When $e_{1}$ lies
close to the $x=y$ line, the order depends on which side $x_{0}$
lies w.r.t. $e_{1}$ which is given by the projection of $x_{0}$
on $e_{2}$.}
\end{figure}
\end{prop}

\section{\label{sec:lds-charge}Linear dynamics in the charging states $s\protect\neq0$}

When $s\neq0$ the system is a linear dynamical system, but the fixed
point is not $0$. The identical system in a charging state $s$ can
be described as
\begin{eqnarray*}
v'(t) & = & B\left[G(s)v(t)-H(s)\right]\\
 & = & BG(s)\left(v(t)-G(s)^{-1}H(s)\right)
\end{eqnarray*}

The fixed point of the system in a state $s$ is 
\[
G(s)^{-1}H(s)=\frac{g_{i}}{g_{s}+g_{i}}s
\]
and the coefficient matrix for the linear flow is 
\[
\left(c_{c}A-(c_{i}+c_{c}n)I\right)^{-1}G(s)=BG(s)
\]
where $B=\left(c_{c}A-(c_{i}+c_{c}n)I\right)^{-1}$ as before (\secref{lds-discharge}).
When $g\gg g_{s}$, i.e. the chargings are much faster than the dischargings,
the fixed points of the system are close to $s$ which are the corners
of the unit cube in $n$ dimensions. Following the arguments as in
\secref{lds-discharge}, even in this case the system trajectory will
converge to an asymptotic direction. The asymptotic ordering of components
would depend on first the fixed point, and in case the fixed point
has equal components then it would also depend on the asymptotic direction
of trajectory. This is explained as:
\begin{prop}
\label{prop:asym-order-fp}In the linear dynamical system of the charging
states $x'(t)=BG(s)\left(x(t)-p\right)$, where $p=\frac{g}{g_{s}+g}s$
is the fixed point and the coefficient matrix $B$ is real, symmetric
and full-rank, the asymptotic permutation of the components will be
same as the permutation of components of the fixed points, i.e. $T(p)$.
In case the fixed point $p$ lies at (or close) to the boundary of
some permutation regions, i.e. some components of $p$ are equal,
the disambiguation of ordering among these components can be done
considering the linear dynamics of $x'(t)=Bx(t)$ with fixed point
shifted to $0$, and following Propositions \ref{prop:asym-order-discharge}.
Hence, the asymptotic order of components is given by
\begin{eqnarray*}
Q_{s}(x_{0}) & = & T(p)\oplus'T(P_{sE_{1}}x_{0})\oplus'T(P_{sE_{2}}x_{0})\oplus'\ldots\\
 & = & T(s)\oplus'T(P_{sE_{1}}x_{0})\oplus'T(P_{sE_{2}}x_{0})\oplus'\ldots
\end{eqnarray*}

where $P_{sE_{1}},P_{sE_{2}},\ldots$ are the projections on the eigenspaces
of $BG(s)$.
\end{prop}
In case the matrix $B$ in the equation $x'(t)=BG(s)\left(x(t)-p\right)$
is not full rank, the system trajectory does not converge to the point
$p$. If $N$ is the null space of the matrix $B$ and $P_{N}$ is
the projection on the null space $N$, then the convergence limit
point for the trajectory starting from $x_{0}$ is is $p+P_{N}x_{0}$.
Also, $N$ is also the null space for $BG(s)$ for all $s$. Hence,
Proposition \ref{prop:asym-order-fp} can be modified for matrices
$B$ which are not full-rank as follows
\begin{prop}
In the linear dynamical system as described in Proposition \ref{prop:asym-order-fp},
but where $B$ is not full rank, the asymptotic order of components
is given by
\[
Q_{s}(x_{0})=T\left(\frac{g_{i}}{g_{i}+g_{s}}s+P_{sN}x_{0}\right)\oplus'T(P_{sE_{1}}x_{0})\oplus'T(P_{sE_{2}}x_{0})\oplus'\ldots
\]

where $P_{sN}$ is the projection matrix on the null space of $BG(s)$.
\end{prop}

When $x_{0}$ is close to the eigenspaces, i.e. magnitude of $P_{sN}x_{0}$
is very small, the additive term of $P_{N}x_{0}$ in the first term
does not change the order determined by $s$. Formally, when $\max\left\{ \left(P_{N}x_{0}\right)_{i}\right\} <\frac{g_{s}}{g_{i}+g_{s}}$
\[
T\left(\frac{g_{i}}{g_{i}+g_{s}}s+P_{N}x_{0}\right)=T(s)\oplus'T(P_{N}x_{0})
\]

and hence,
\begin{equation}
Q_{s}(x_{0})=T\left(s\right)\oplus'T(P_{sN}x_{0})\oplus'T(P_{sE_{1}}x_{0})\oplus'T(P_{sE_{2}}x_{0})\oplus'\ldots\label{eq:Qs-condition}
\end{equation}

\begin{figure}
\begin{centering}
\includegraphics{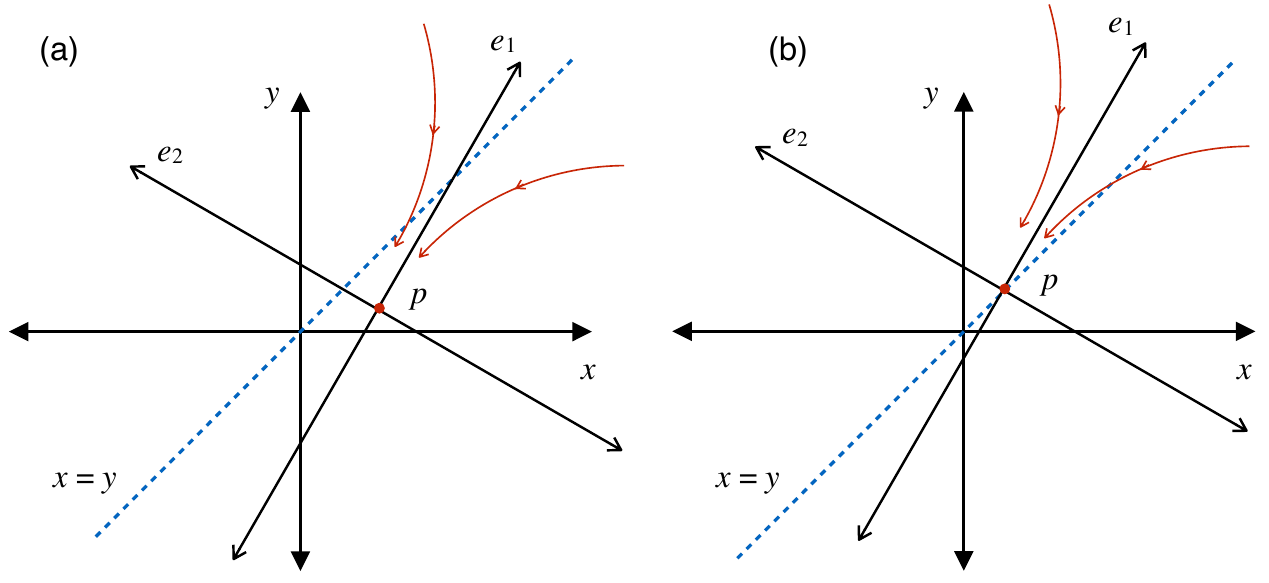}
\par\end{centering}
\caption{(a) When the fixed point in a two dimensional linear dynamical system
is not $0$ then the asymptotic order of the components is determined
by the fixed point $p$. (b) If the fixed point lies on the $x=y$
line, which is a boundary of permutations regions, then the disambiguation
is done using the eigenvectors.}
\end{figure}

\subsection{Approximation by instantaneous charging}

If the chargings are very fast, i.e. $\frac{g_{s}}{g_{i}}\rightarrow0$,
we can approximate the chargings by an instantaneous change in the
state from $x$ to $x+\Delta x$ by linearizing the system at the
time instant when the state changes from $s=0$ to the charging state.
Let $\hat{S}$ denote a diagonal matrix such that $diag(\hat{S})=s$
where $s$ is the state vector. When $s\neq0$ we have from (\ref{eq:system-identical})
\begin{eqnarray*}
x'(t) & = & B\left(g_{i}\hat{S}x+g_{s}\left(s-x\right)\right)\\
 & = & Bg_{i}\left(\hat{S}x+\frac{g_{s}}{g_{i}}(s-x)\right)\\
 & \simeq & g_{i}B\hat{S}x
\end{eqnarray*}

If the $k^{th}$ node charges then $\hat{S}x=v_{l}e_{k}$ where $e_{k}$
is $k-axis$ vector whose all components are $0$ expect the $k^{th}$
which is $1$. If the $k^{th}$ node charges completely from $v_{l}$
to $v_{h}$ without any state transition in between, we have 
\begin{eqnarray*}
\left(\Delta x\right)_{k} & = & dv\\
\implies(x')_{k}\Delta t & = & dv\\
\implies\Delta t & = & \frac{dv}{(g_{i}B\hat{S}x)_{k}}\\
 & = & \frac{dv}{g_{i}v_{l}e_{k}^{T}Be_{k}}\\
 & = & \frac{dv}{g_{i}v_{l}B_{kk}}
\end{eqnarray*}

Therefore,
\begin{eqnarray*}
\Delta x & = & x'\Delta t\\
 & = & dv\frac{v_{l}g_{i}Be_{k}}{v_{l}g_{i}B_{kk}}\\
 & = & \frac{dv}{B_{kk}}Be_{k}
\end{eqnarray*}

which is just a scaled column vector of $B$. We have the following:
\begin{prop}
\label{prop:instantaneous-charging}In the dynamical system of (\ref{eq:system-identical}),
when $s\neq0$ and only a single node charges, the chargings can be
approximated by linearizing the system. If the transition occurs from
$x$ to $x+\Delta x$ then $\Delta x$ is given by:
\[
\Delta x=\frac{dv}{B_{kk}}Be_{k}
\]
\end{prop}
\begin{rem}
\label{rem:independent-delta-x}An important point to note here is
that this change is independent of $x$.
\end{rem}

\section{\label{sec:color-sorting}Vertex Color-Sorting}

As can be seen in the system equation of the capacitively coupled
oscillators, the discharge phase (where all oscillators are discharging)
is a simple linear differential equation with $H(s)=0$. The matrix
$C-C_{C}$ is just the Laplacian matrix of the graph of the oscillators
and the system dynamics is governed by simply the eigenspectrum of
the of the Laplacian matrix of the graph. As such, there are interesting
connections between spectral algorithms for graph coloring and the
coupled relaxation oscillator circuit.
\begin{defn}
(k-Color-Sorting) An ordering $u=\{u_{i}\}$, $i\in[1,n]$ of the
$n$ nodes of a graph is a proper k-Color-Sorting if there exists
a proper k-Coloring $\{c_{i}\}$, $i\in[1,n]$, where $c_{i}$ is
the color assigned to the $i^{th}$ node such that all nodes with
the same color appear together in $u$, i.e. for any nodes $i,j,k$
with $u_{i}<u_{k}<u_{j}$, $c_{i}=c_{j}\implies c_{i}=c_{k}=c_{j}$.
This can be extended to a cyclic ordering where the nodes with the
same color appear together.
\end{defn}
\begin{lem}
For a graph with $n$ nodes, adjacency matrix $A$ and chromatic number
$\chi_{A}$:
\end{lem}
\begin{enumerate}
\item Any ordering of nodes $S$ is a proper k-Color-Sorting for some $k$
such that $\chi_{A}\leq k\leq n$. 
\item Let $B(M)$ be the minimum number of diagonal blocks which are identically
$'0'$ and which cover the complete diagonal of the matrix $M$. The
minimum $k$ for which $S$ is a proper k-Color-Sorting is $B(PAP^{T})$.
If $S$ is a proper k-Color-Sorting and $P$ its permutation matrix,
then 
\[
\chi_{A}\leq B(PAP^{T})\leq k
\]
\end{enumerate}
\begin{proof}
Any ordering $S$ is a proper n-Color-Sorting, and if $S$ is a proper
k color sorting then minimum number of colors can be $\chi_{A}$. 

If $P$ is the permutation matrix of an ordering $u$, then $PAP^{T}$
is the adjacency matrix of a graph with the ordering of nodes changed
to $u$. If $u$ is a proper k-Color-Sorting then, $PAP^{T}$ will
have at least $k$ number of $'0'$ diagonal blocks, one corresponding
to each color group, hence, $B(PAP^{T})\leq k$. Also, the diagonal
blocks which are $'0'$ also determine a valid coloring of the graph
and hence $B(PAP^{T})\geq\chi_{A}$. 
\end{proof}
\begin{prop}
For a k-chromatic graph, k-Color-Sorting is NP hard. Moreover, finding
the chromatic number $\chi_{A}$ of a graph with adjacency matrix
$A$ and the proper $\chi_{A}$-Coloring is equivalent to the following
optimization problem:
\[
\begin{array}{cc}
min\,B(PAP^{T}), & P\in all\,permutations\,of\,nodes\end{array}
\]
where the solution P is a proper $\chi_{A}$-Color-Sorting, $\chi_{A}=min\{B(PAP^{T})\}$.
\end{prop}
\begin{proof}
Computing $B(PAP^{T})$ is a $O(n^{2})$ problem, n being the number
of nodes because there are $n^{2}$ elements in $PAP^{T}$. And for
a k-chromatic graph, $\chi_{A}=B(PAP^{T})=k$ where $P$ is a proper
k-Color-Sorting. Hence, $\chi_{A}$ can be computed in $O(n^{2})$
if a proper k-Color-Sorting $P$ can be found. 

Also, for any permutation $P$, $B(PAP^{T})\geq\chi_{A}$ as stated
above, where equality holds only when $P$ is a proper $\chi_{A}$-Color-Sorting.
Hence, finding chromatic number is equivalent to the stated optimization
problem. Also, once a proper $\chi_{A}$-Color-Sorting is known, the
$'0'$ diagonal blocks also determine the proper $\chi_{A}$-Coloring.
\end{proof}

\section{Cycles in the prototypical case: complete graphs with equal nodes
in each class}

Using the results in the previous sections, we can understand why
a cycle would exist in the prototypical case of a complete graph when
the number of nodes in each class is equal. 
\begin{prop}
\label{prop:cycle-conditions}The following three conditions when
satisfied result in the existence of a cycle and helps us understand
why the possibility of it reduces as graphs become sparser, and hence
harder.

\begin{enumerate}
\item \textbf{Attractor:} The system in state $s=0$ tries to order the
components of the state vector in the correct vertex color-sorting.
Hence, if the system starts from a state $x_{0}$ whose order of components
is same as the final asymptotic order, i.e. $T(x_{0})=Q_{0}(x_{0})$,
then with time $T(x(t))$ remains constant.
\item \textbf{Ordering: }The charging spikes just change the order of components
of $x$ by a circular permutation. If the $k^{th}$ oscillator charges
from $v_{l}$ to $v_{h}$ then the order of all other components remains
same.
\item \textbf{Sustaining the cycle:} If condition 2 is true then the charging
transitions cycle the order of $x_{0}$ to all the circular permutations.
For a cycle to exist, the state $s=0$ should not only preserve the
order of $x_{0}$ when $T(x_{0})=Q_{0}(x_{0})$ but it should also
have lower tendency to change the order when $T(x_{0})$ is any circular
permutation of $Q_{0}(x_{0})$. 
\end{enumerate}
\end{prop}
Why these conditions hold in the prototypical case of complete graph
with equal number of nodes in each color class can be seen as follows. 

\textit{Explanation for condition 1: }The adjacency matrix $A$ in
the prototypical case is a low rank matrix with the rank equal to
the number of colors, i.e. if it is a $k$-partite graph then rank
is n. The adjacency matrix is a block matrix with equal sized $k^{2}$
blocks and the diagonal blocks are $0$ and the non-diagonal blocks
are $1$. One eigenvector of the matrix $A$ is the constant vector
$[1,1,1,...]$ which is the diagonal of the n-dimensional cube $[v_{l},v_{h}]^{n}$
and also lies at the intersection of all the simplexes of the permutation
regions (equation \ref{eq:permutation-regions}) and does not affect
the asymptotic order of components of $x$. Hence all the other eigenvectors
decide the asymptotic order and lie in the non-positive quadrants.
The eigenvectors of $B$ with least negative eigenvalues (which are
the eigenvectors of A with most negative eigenvalues) have components
which are equal on each color class (Appendix \ref{subsec:B-eigenvectors})
and hence should direct the system towards a correct vertex color-sorting
in state $s=0$. We also know that all the eigenvalues of the coefficient
matrix in the state $s=0$ are negative, and hence, if the system
starts with the correct order of components, i.e. $T(x_{0})=Q_{0}(x_{0})$
then the system state $x$ will continue to lie in the same permutation
region with time.

\textit{Explanation for condition 2: }Assuming very fast charging
and using the instantaneous charging approximation, we see from Proposition
\ref{prop:instantaneous-charging} that the state transition $\Delta x$
is in the direction of the $k^{th}$ column vector of $B$ when the
$k^{th}$ node charges. As shown in appendices \ref{subsec:B-structure}
and \ref{subsec:B-column}, in case of weak coupling, i.e. $c_{i}\gg c_{c}$
the $k^{th}$ column vector is constant for all non-charging components
and hence $\Delta x$ does not change the order of the non-charging
components. The variation in the non-charging components of $\Delta x$
is inversely propotional to $n+m$ and hence with larger $n$ and
$m$ the charging transition $x\rightarrow x+\Delta x$ tries to preserve
the order of non-charging components more (Appendix \ref{subsec:B-column}).
As shown in figure \ref{fig:weak_coupling} the effect of charging
transitions can be seen as small kinks in the waveforms of non-charging
components. The magnitude of these kinks is negligible for weak coupling
(a), and is clearly visible for stronger coupling (c). Even though
the charging transitions affect the non-charging components in the
case of a stronger coupling, the order of non-charging components
is not disturbed, i.e. the change in all the non-charging components
is almost the same (Appendix \ref{subsec:B-column}).

\textit{Explanation for condition 3: }If the system state $x$ is
close to the eigenspace of $B$ with least negative eigenvalue, say
$E_{1}$, then $x$ has components which are close for the same color
class (Appendix \ref{subsec:B-eigenvectors}) and components of different
color classes will have more separation between them by comparison.
If the components of $x$ are ordered in increasing order then it
will have a pattern $\{x_{a_{1}},x_{a_{2}},...,x_{b_{1}},x_{b_{2}},...,x_{c_{1}},x_{c_{2}},...\}$,
where $a_{i}$ are the indices for one color class, $b_{i}$ for another
etc. If the order among the color classes is changed, say $\{x_{b_{1}},x_{b_{2}},...,x_{a_{1}},x_{a_{2}},...,x_{c_{1}},x_{c_{2}},...\}$
even then $x$ will be close to the eigenspace $E_{1}$ because of
the multiplicity of the least negative eigenvalue (Appendix \ref{subsec:B-eigenvectors}).
The charging transitions of nodes of the same color class will occur
consecutively with little time durations between them. This little
time does not allow the system state $s=0$ which occurs between these
transitions to change the order. When all nodes of one particular
class have undergone charging processes, the system state $x$ again
comes close to the eigenspace $E_{1}$ because the components of $x$
belonging to the same color class are again close to each other. Hence,
the state $s=0$ does not disurb this order as well. The cycle repeats
with very fast consecutive charging processes of the next color class.
This also gives rise to clustering of the phases of nodes w.r.t. their
color classes. 

\begin{figure}
\begin{centering}
\includegraphics{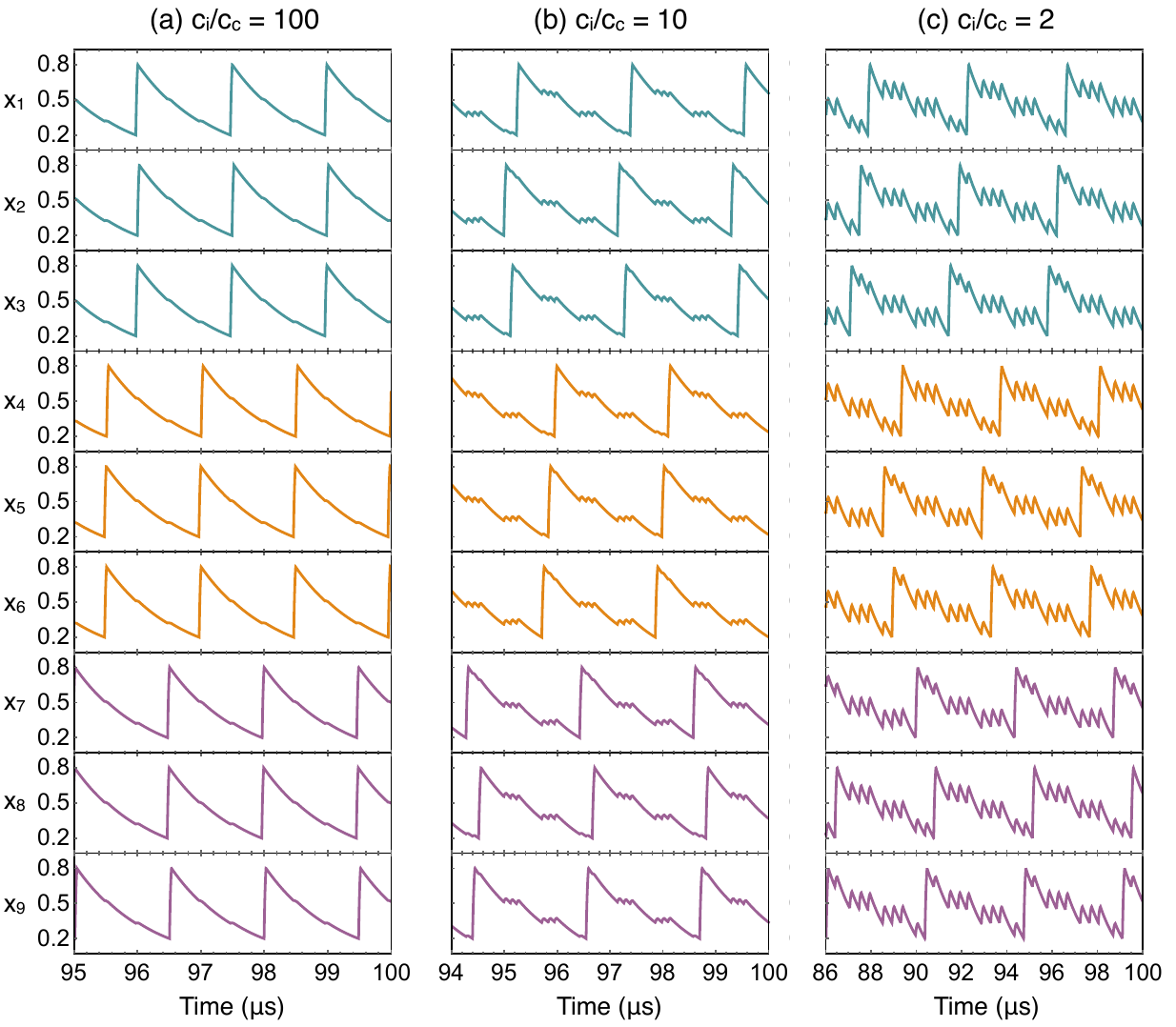}
\par\end{centering}
\caption{Simulation waveforms of a coupled relaxation oscillator circuit connected
in a complete 3-partite graph with 3 nodes in each color class for
different $c_{i}/c_{c}$ values (a) 100, (b) 10, and (c) 2. As can
be seen, the charging transitions do not affect the non-charging components
of the state vector $x$in case of weak coupling (a). In case of stronger
coupling (c), even though the charging transitions affect the non-charging
components (seen as small kinks in the waveforms), the order of non-charging
components is undisturbed as discussed in Appendix \ref{subsec:B-column}.
\label{fig:weak_coupling}}
\end{figure}

\section{Cycles in the general case}

Adjacency matrices of non-simple graphs can be considered as perturbations
to the prototypical cases of complete graphs, and using perturbation
theory of matrices we can say that the eigenvectors of perturbed matrices
are rotations of the original eigenvectors \citep{davis_rotation_1963},
where the extent of rotation depend on the amount of perturbation.
Hence, even in non-simple cases, the eigenvectors with most negative
eigenvalues of the adjacency matrix will tend to have components which
are close to each other within the same color class and away from
those of different color classes. This property has been explored
with mathematical detail in works related to spectral algorithms for
graph coloring \citep{Alon:aa,Aspvall:aa,McSherry:aa}. When viewed
from the perspective of a coupled relaxation oscillator system of
(\ref{eq:system-identical}), the above mentioned property of eigenvectors
of the adjacency matrix $A$ with most negative eigenvalues will be
shared by the eigenvectors of $B$ with the least negative eigenvalues
because of Proposition \ref{prop:eigenvalues-B-A}. As shown above,
the asymptotic order of components of the system state in the discharge
phase $s=0$ of coupled relaxation oscillator systems depend on the
least negative eigenvalues of $B$. Hence, the relaxation oscillator
systems in state $s=0$ is expected to direct the system towards correct
vertex color sorting, which satisfies condition 1 of Proposition \ref{prop:cycle-conditions}.
Conditions 2 and 3 also depend on eigenvectors and hence similar arguments
of matrix perturbation can be applied.

\section{Prototypical experiments and validation}

Vanadium dioxide (VO\textsubscript{2}) is a prototypical insulator-metal
transition material system with strong electron-electron and electron-phonon
interactions that has been the subject of intense fundamental and
applied research . The above room temperature phase transition (transition
temperature = 340 K) in VO\textsubscript{2} has an electronic component
characterized by an abrupt change in resistivity (and carrier concentration)
up to five orders in magnitude; the large increase in carrier concentration
can be attributed to collapse of the 0.6 eV band gap (optically measured)
across the insulator-to-metal transition. Further, the phase transition
also has a structural component wherein the crystal structure evolves
from the monoclinic M1 phase with dimerized vanadium atoms in the
low-temperature insulating state to rutile crystal structure in the
high-temperature metallic phase.

\begin{figure}
\begin{centering}
\includegraphics{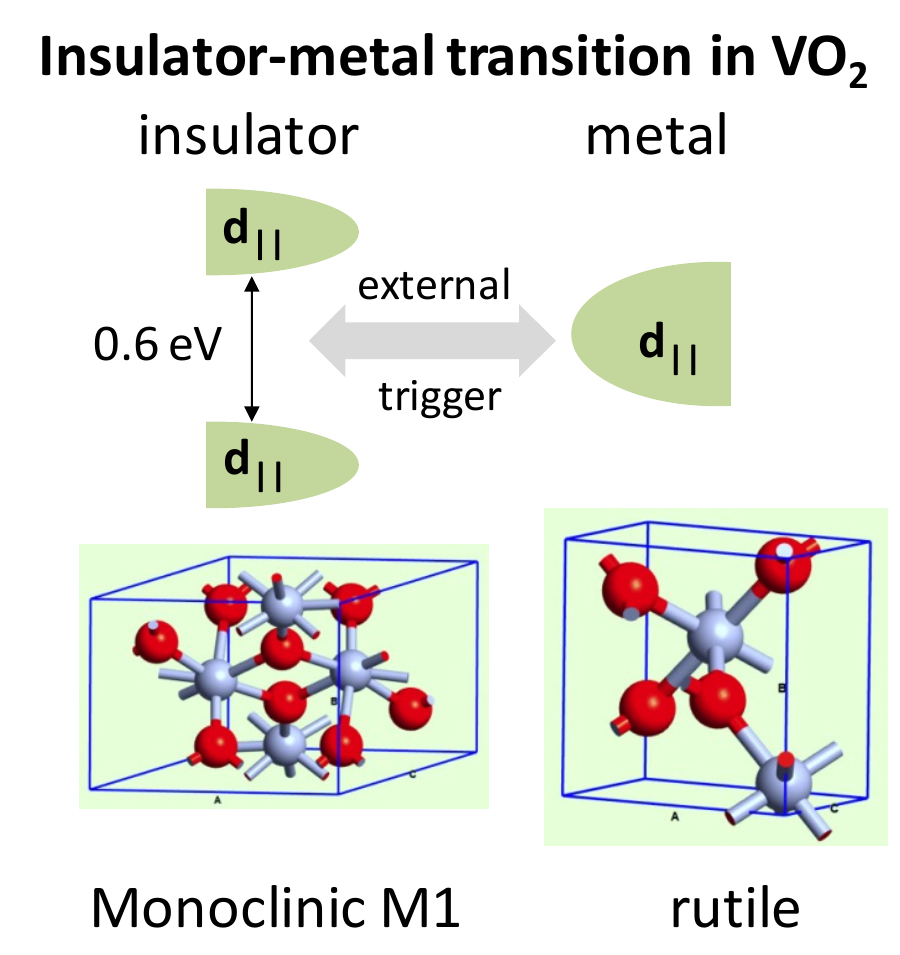}
\par\end{centering}
\caption{Insulator-metal transition in VO\protect\textsubscript{2} showing
phase change}
\end{figure}

Despite intense research efforts, the origin of the phase transition
in VO\textsubscript{2} has been a subject of debate with competing
theories suggesting that the driving force behind the transition could
be Mott or Peierl\textquoteright s physics as well as a weighted combination
of both the mechanisms. Further, the electrically induced phase transition
is VO\textsubscript{2} which is relevant to electronic VO\textsubscript{2}
devices like the relaxation oscillators discussed here, is debated
to be carrier density driven or of electro-thermal nature.

With respect to the relaxation oscillators discussed here, the unknown
nature of origin of the electrically induced phase transition in VO\textsubscript{2}
entails that the critical voltage ($V_{h}$, $V_{l}$ in figure 
of main text)/ current cannot be quantitatively predicted even though
empirically measurements indicate that the typical critical electric
field values are in the 20-60 kV/cm range. However, we emphasize that
knowing $V_{l}$ and $V_{h}$, the oscillators can de designed in
a deterministic manner.

The details of the experiments, experimental conditions and the theory
connecting experiments with linear dynamical systems for the case
of a single and a coupled pair of oscillators can be found in the
authors' earlier publications in \citep{parihar2015synchronization}.

\appendix

\part*{Appendix}

\section{The coefficient matrix in prototypical case}

In this section we give an analytical treatment of the structure of
the coefficient matrix and its eigen spectrum in the prototypical
case. We consider the prototypical case where the graph is complete
and the number of nodes in each color class is equal. When $n$ identical
oscillators with internal capacitances $c_{i}$ are connected in a
$k$-partite graph, and the coupling is purely capacitive with same
coupling capacitances $c_{c}$ used for all pairs, then the system
evolution is described as in equation \ref{eq:system-identical}.
In the simple case when each partition has equal number of nodes $m=n/k$,
then more can be said about the coefficient matrix $B=(c_{i}I-c_{c}A+c_{c}nI)^{-1}$.
Let $F=(c_{i}I-c_{c}A+c_{c}nI)^{-1}$ so that $B=F^{-1}$. Then $F$
can be written as a repeated partitioned matrix as
\[
F=U\otimes G+V\otimes E
\]
where $\otimes$ is the kronecker product of matrices, $U$ and $V$
are $k\times k$ matrices, $G$ and $E$ are $m\times m$ matrices,
and the matrices are given by 
\begin{eqnarray*}
U & = & c_{ic}I_{k}\\
G & = & I_{m}\\
V & = & I_{k}-J_{k}\\
E & = & c_{c}J_{m}
\end{eqnarray*}
with $I_{m}$ being the $m\times m$ identity matrix, $J_{m}$ the
$m\times m$ matrix with all ones, and $c_{ic}=\left(c_{i}+nc_{c}\right)$.

\subsection{\label{subsec:B-eigenvectors}Eigenvectors of $B$ in prototypical
case}

For $n$ nodes and $k$ color classes, let $U$ be a $n\times m$
matrix where each column vector corresponds to one color class where
the components of that particular class are $k/n$ and rest are $0$.
As such, $U^{T}AU$ is a $k\times k$ matrix with each entry equal
to the average of entries of the corresponding block in $A$. In the
simple case of complete graph with equal number of nodes in each class,
$U^{T}AU=J-I$ where $J$ is a square matrix of all ones and $I$
is the identity matrix. If $x$ is an eigenvector of $U^{T}AU$ then
\begin{eqnarray*}
U^{T}AUx & = & \lambda x\\
UU^{T}A(Ux) & = & \lambda(Ux)
\end{eqnarray*}

Now $UU^{T}A$ is just the scaled version of $A$ and hence, 
\[
\alpha A(Ux)=\lambda(Ux)
\]

Therefore if $x$ is an eigenvector of $U^{T}AU$ then $Ux$ is an
eigenvector of $A$. Also the number of non-zero eigenvalues of $A$
are $k$ which is equal to the rank of $U^{T}AU$ which is full-rank.
Hence all the eigenvectors of $A$ can be described using the eigenvectors
of $U^{T}AU$ and they have equal components in a single color class.
$J-I$ has an eigenvalue $-1$ with multiplicity $n-1$, and an eigenvalue
$n-1$, and so does $A$. Now the eigenvectors of $B$ with the least
negative eigenvalues are same as that of $A$ with most negative eigenvalues
(Proposition \ref{prop:eigenvalues-B-A}). Hence, the eigenvalues
of $B$ with least negative eigenvalues are constant on each color
class.

\subsection{\label{subsec:B-structure}Structure of the inverse of $F$ in prototypical
case}
\begin{prop*}
If $F=(c_{i}I-c_{c}A+c_{c}nI)$ is the coefficient matrix of the network,
then $B=F^{-1}$ has the same partitioned form as $F$. More precisely,
$B=F^{-1}$ can be written as
\[
F^{-1}=\frac{1}{c_{ic}}\left(\frac{1}{c_{ic}}U\otimes G+D\otimes E\right)
\]
where $U$, $G$ and $E$ are the same matrices that describe $F$,
$c_{ic}$ is as defined above, and D is a $k\times k$ matrix given
by
\begin{eqnarray*}
D & = & \frac{1}{c_{i}+(n+m)c_{c}}\left(\beta J_{k}-I_{k}\right)
\end{eqnarray*}

and
\[
\beta=\frac{c_{i}+nc_{c}}{c_{i}+mc_{c}}
\]
\end{prop*}
\begin{proof}
As described above, $F=U\otimes G+V\otimes E$. Here $G$ is a identity
matrix and $E$ is a rank 1 matrix. Hence, as shown in \citep{1981},
the inverse for $F$ can be calculated as
\[
F^{-1}=U^{-1}\otimes G-\left[U+\left(\mathrm{tr\,}E\right)V\right]^{-1}VU^{-1}\otimes E
\]

Now,
\begin{eqnarray*}
\mathrm{tr\,}E & = & mc_{c}\\
U^{-1} & = & \frac{1}{c_{ic}}I_{k}\\
VU^{-1} & = & \frac{1}{c_{ic}}\left(I_{k}-J_{k}\right)\\
\left[U+\left(\mathrm{tr}E\right)V\right]^{-1} & = & \left[U+mc_{c}V\right]^{-1}\\
 & = & \left[\left(c_{i}+(n+m)c_{c}\right)I_{k}-mc_{c}J_{k}\right]^{-1}\\
 & := & \left[P-Q\right]^{-1}
\end{eqnarray*}

As $Q$ is a rank 1 matrix, we can use another result from \citep{1981}:
\begin{eqnarray*}
\left[U+\left(\mathrm{tr}\,E\right)V\right]^{-1} & = & \left[P-Q\right]^{-1}\\
 & = & P^{-1}+\frac{1}{1-\mathrm{tr}\,QP^{-1}}P^{-1}QP^{-1}\\
 & = & \frac{1}{c_{i}+(n+m)c_{c}}I_{k}+\frac{1}{1-\frac{nc_{c}}{c_{i}+(n+m)c_{c}}}\frac{1}{\left(c_{i}+(n+m)c_{c}\right)^{2}}mc_{c}J_{k}\\
 & = & \frac{1}{c_{i}+(n+m)c_{c}}\left(I_{k}+\frac{mc_{c}}{c_{i}+mc_{c}}J_{k}\right)
\end{eqnarray*}

Combining the parts, and noting that $J_{k}^{2}=kJ_{k}$, we get
\begin{eqnarray*}
\left[U+\left(\mathrm{tr}E\right)V\right]^{-1}VU^{-1} & = & \frac{1}{c_{i}+(n+m)c_{c}}\left(I_{k}+\frac{mc_{c}}{c_{i}+mc_{c}}J_{k}\right)\frac{1}{c_{ic}}\left(I_{k}-J_{k}\right)\\
 & = & \frac{1}{c_{ic}\left(c_{i}+(n+m)c_{c}\right)}\left(I_{k}-\beta J_{k}\right)
\end{eqnarray*}

where,
\[
\beta=\frac{c_{i}+nc_{c}}{c_{i}+mc_{c}}
\]

Finally,
\[
F^{-1}=\frac{1}{c_{ic}}\left[I_{k}\otimes I_{m}+\frac{1}{c_{i}+(n+m)c_{c}}\left(\beta J_{k}-I_{k}\right)\otimes c_{c}J_{m}\right]
\]

and hence,
\begin{equation}
B=F^{-1}=\frac{1}{c_{ic}}\left(\frac{1}{c_{ic}}U\otimes G+D\otimes E\right)\label{eq:F-inverse}
\end{equation}
\end{proof}

\subsection{\label{subsec:B-column}Column vector of $B$ in prototypical case}

Using equation \ref{eq:F-inverse} we can deduce properties of the
column vector of $B$.
\begin{prop*}
Let $B_{k}$ be the $k^{th}$ column vector of $B$ and $B_{kl}$
be the $(k,l)^{th}$ element of $B$. For the components of $B_{k}$
there are only 3 kinds of values.

\begin{enumerate}
\item For the $k^{th}$ component, 
\[
B_{kk}=\frac{1}{c_{ic}}\left(1+\alpha(\beta-1)\right)
\]
\item For all other components in the same class as the $k^{th}$ component,
i.e. when $k^{th}$ and $l^{th}$ node are in the same color class
\[
B_{kl}=\frac{1}{c_{ic}}\alpha(\beta-1)
\]
\item For all other components of $B_{k}$ which are not in the same partition/color
class as the $k^{th}$ node, i.e. when $k^{th}$ and $j^{th}$ node
are not in the same class
\[
B_{kj}=\frac{1}{c_{ic}}\alpha\beta
\]
where 
\[
\alpha=\frac{c_{c}}{c_{i}+(n+m)c_{c}}
\]
\item The difference between $B_{kl}$ and $B_{kj}$ w.r.t. $B_{kk}$ is
given by:
\[
\frac{B_{kj}-B_{kl}}{B_{kk}}=\frac{1}{r+n+m+\frac{n-m}{r+m}}
\]
where $r=c_{i}/c_{c}$. As can be seen, this difference can be made
very small by weak coupling, i.e. $c_{c}\ll c_{i}$, but more importantly
for increasing $n$ and $m$ this difference reduces
\end{enumerate}
\end{prop*}
\bibliographystyle{unsrt}
\bibliography{refs_supp}

\end{document}